\documentclass{CSML}

\def\dOi{12(2:8)2016}
\lmcsheading%
{\dOi}
{1--31}
{}
{}
{Nov.~\phantom02, 2015}
{Jun.~23, 2016}
{}

\ACMCCS{[{\bf Theory of computation}]: Logic}
\subjclass{F.4.1 "Mathematical Logic"}

\pdfoutput=1
\usepackage{hyperref}
\usepackage{proof} %dla \infer i \deduce

\title[Two-variable Logic with Counting and a Linear Order]{Two-variable Logic with~Counting~and~a~Linear~Order}
\thanks{Research supported by NCN Grant no.  2011/03/B/ST6/00346}

\author[W.~Charatonik]{Witold Charatonik}
\address{Institute of Computer Science, University of Wroclaw, Poland}
\email{Witold.Charatonik@cs.uni.wroc.pl, pwit@pwit.info}

\author[P.~Witkowski]{Piotr Witkowski}
\address{\vspace{-18 pt}}%Institute of Computer Science, University of Wroclaw, Poland}
\keywords{Two-variable logic, counting quantifiers, linear order, satisfiability, complexity}

\usepackage{stmaryrd} %for \llbracket

\usepackage{xspace}
\usepackage{amsmath, amsthm, amssymb, latexsym}

\usepackage{algorithmic}

\begin{document}
\newcommand{\Allow}{\mathit{Allowed}\xspace}
\newcommand{\Visit}{\mathit{Visited}\xspace}

\newcommand{\MPending}{\overrightarrow{\mathsf{Processed}}\xspace}
\newcommand{\UPending}{\overrightarrow{\mathsf{Cut}}\xspace}

\newcommand{\UpArrow}[1]{{#1}^<}
\newcommand{\DownArrow}[1]{{#1}^>}

\newcommand{\Kings}{K\xspace}
\newcommand{\Nat}{\mathbb{N}\xspace}
\newcommand{\Zet}{\mathbb{Z}\xspace}
\newcommand{\ST}{\mathrm{ST}\xspace}
\newcommand{\Partial}[1]{\mathit{partial}(#1)\xspace}
\newcommand{\tpf}[1]{\mathrm{tp}_1(#1)\xspace}
\newcommand{\tps}[1]{\mathrm{tp}_2(#1)\xspace}
\newcommand{\Guess}{\textbf{guess}\xspace}
\newcommand{\with}{\textbf{with}\xspace}
\newcommand{\Reject}{\textbf{Reject}\xspace}

\newcommand{\inc}{\textbf{inc}\xspace}
\newcommand{\dec}{\textbf{dec}\xspace}

\newcommand{\IIf}{\textbf{if}\xspace}
\newcommand{\IThen}{\textbf{then}\xspace}
\newcommand{\IElse}{\textbf{else}\xspace}
\newcommand{\IEndIf}{\textbf{endif}\xspace}

\newcommand{\IWhile}{\textbf{while}\xspace}
\newcommand{\IDo}{\textbf{do}\xspace}
\newcommand{\IEndWhile}{\textbf{endwhile}\xspace}

\newcommand{\IForAll}{\textbf{for all}\xspace}
\newcommand{\IEndFor}{\textbf{end for}\xspace}

\newcommand{\State}{\textbf{State}\xspace}
\newcommand{\StState}{\textbf{Initial State}\xspace}
\newcommand{\AccState}{\textbf{Accepting State}\xspace}
\newcommand{\Transition}{\textbf{Transition}\xspace}

%frame
\newcommand{\Frame}{{\mathcal{F}}}

%class of finite structures with a linear order and its successor
\newcommand{\ClassO}{{\mathcal{O}}(\Sigma,\Lin,\Succ)}

%complexity classes
\newcommand{\NEXPTIME}{\textsc{NExpTime}}
\newcommand{\EXPTIME}{\textsc{ExpTime}}
\newcommand{\EXPSPACE}{\textsc{ExpSpace}}

%Two-variable logic
\newcommand{\FOTwo}{$\mathrm{FO}^2$\xspace}

%Two-variable logic with counting
\newcommand{\CTwo}{$\mathrm{C}^2$\xspace}

%One unnamed order and successor
\newcommand{\Succ}{{+\!\!1}}
\newcommand{\Pred}{-\!1}
\newcommand{\Lin}{<}
\newcommand{\LinEq}{\leq}

\newcommand{\true}{\mathbf{true}}
\newcommand{\false}{\mathbf{false}}
\newcommand{\oor}{\mathbf{or}}
\newcommand{\aand}{\mathbf{and}}

\newcommand{\Var}{\mathrm{V}_{\mathit{fin}}}
\newcommand{\Vect}{\mathrm{Vec}_{\Nat}}
\newcommand{\Type}{\mathrm{Type}}
\newcommand{\Expr}{\mathit{Expr}}
\newcommand{\Test}{\mathit{Test}}

\newcommand{\sem}[3]{\left({#1},{#2}\right)\leadsto{#3}}

\newcommand{\CTwoGen}[3]{\mbox{$\mathrm{C}^2\left(#1,#2,#3\right)$}\xspace}

\newcommand{\CTwoGenSimple}[2]{\mbox{$\mathrm{C}^2\left(#1,#2\right)$}\xspace}

\newcommand{\CTwoLin}[1]{\mbox{$\mathrm{C}^2\left(#1\right)$}\xspace}

\newcommand{\FOTwoLin}[1]{\mbox{$\mathrm{FO}^2\left(#1\right)$}\xspace}

%Two named orders and successors
\newcommand{\fSucc}{\Succ_{1}}
\newcommand{\fLin}{<_{1}}
\newcommand{\sSucc}{\Succ_{2}}
\newcommand{\sLin}{<_{2}}

\newcommand{\CTwoOneLin}{\CTwoLin{*,*,\{\Lin\}}}

\newcommand{\CTwoOneLinPlus}{\CTwoLin{*,*,\{\Lin\}}}

\newcommand{\CTwoOneLinS}{\CTwoLin{*,*,\{\Lin,\Succ\}}}

\newcommand{\CTwoLinUndecid}{\CTwoLin{*,2,\{\fLin,\sLin\}}}

\newcommand{\CTwoLinUndecidS}{\CTwoLin{*,0,\{\fLin,\fSucc,\sLin,\sSucc\}}}

\newcommand{\CTwoLinUndecidA}{\CTwoLin{*,*,\{\fLin,\sLin\}}}

\newcommand{\FOTwoLinUndecidS}{\FOTwoLin{*,0, \{\fLin,\fSucc,\sLin,\sSucc\}}}

\newcommand{\CTwoOnePrec}{\CTwoLin{*,*,\{\prec\}}}
\newcommand{\CTwoTwoPrec}{\CTwoLin{*,2,\{\prec\}}}

\newcommand{\CTwoLinUnknownA}{\CTwoLin{*,0,\{\fLin,\sLin\}}}

\newcommand{\CTwoLinUnknownB}{\CTwoLin{*,1,\{\fLin,\sLin\}}}

\newcommand{\CTwoLinSucc}{\CTwoLin{*,0,\{\Lin,\Succ\}}}

\newcommand{\FOTwoLinSucc}{\FOTwoLin{*,0,\{\Lin,\Succ\}}}

\newcommand{\CTwoSBin}{\CTwoLin{*,1,\{\Lin,\Succ\}}}

\newcommand{\CTwoLinOne}{\CTwoLin{*,1,\{\Lin\}}}

% abbreviations
\newcommand{\ie}{i.\,e.,\xspace}
\newcommand{\Ie}{I.\,e.,\xspace}
\newcommand{\eg}{e.\,g.,\xspace}
\newcommand{\Eg}{E.\,g.,\xspace}
\newcommand{\cf}{cf.\xspace}
\newcommand{\Cf}{Cf.\xspace}
\newcommand{\wrt}{w.r.t.\ }
\newcommand{\resp}{resp.\xspace}
\newcommand{\Wrt}{W.r.t.\xspace}
\newcommand{\wlg}{w.l.o.g.\xspace}
\newcommand{\Wlg}{W.l.o.g.\xspace}

\newcommand{\tuple}[1]{\langle#1\rangle} 
\newcommand{\tpu}{\mathrm{tp}^{\str A}}
\newcommand{\tp}[1]{\mathrm{tp}^{#1}}

\newcommand{\forwardtp}[1]{\mathrm{tp}_{>}^{#1}}
\newcommand{\backwardtp}[1]{\mathrm{tp}_{<}^{#1}}

\newcommand{\otp}[1]{\mathrm{tp}_{#1}}
\newcommand{\stp}[1]{\mathrm{st}^{#1}}

\newcommand{\K}{{\mathcal{K}}}
\newcommand{\M}{{\mathcal{M}}}
\newcommand{\sigmacphigraph}{G_{\Sigma,c}^\varphi}
\newcommand{\sigmacgraph}{G_{\Sigma,c}}

\newcommand{\str}[1]{{\mathcal{#1}}}

\theoremstyle{definition}
\newtheorem{proposition}[thm]{Proposition}
\newtheorem{definition}[thm]{Definition}
\newtheorem{remark}[thm]{Remark}
\theoremstyle{plain}
\newtheorem{theorem}[thm]{Theorem}
\newtheorem{lemma}[thm]{Lemma}
\newtheorem{corollary}[thm]{Corollary}

\newcommand{\tsucc}{\theta_{\Succ}}
\newcommand{\tprec}{\theta_{\Pred}}
\newcommand{\tless}{\theta_{<}}
\newcommand{\tgreat}{\theta_{>}}
\newcommand{\teq}{\theta_{=}}
\newcommand{\trans}{\mathit{trans}}
\newcommand{\conf}{\mathit{conf}}

\begin{abstract}
  We study the finite satisfiability problem for the two-variable
  fragment of first-order logic extended with counting quantifiers
  (\CTwo) and interpreted over linearly ordered structures. We show
  that the problem is undecidable in the case of two linear orders (in
  the presence of two other binary symbols). In the case of one linear
  order it is \NEXPTIME-complete, even in the presence of the
  successor relation. Surprisingly, the complexity of the problem
  explodes when we add one binary symbol more: \CTwo with one linear
  order and in the presence of other binary predicate symbols is
  equivalent, under elementary reductions, to the emptiness problem
  for multicounter automata.
\end{abstract}

\maketitle

\section{Introduction}
Since 1930s, when Alonzo Church and Alan Turing proved that the
satisfiability problem for first-order logic is undecidable, much
effort was put to find decidable subclasses of this logic. One of the
most prominent decidable cases is the two-variable fragment
\FOTwo.  \FOTwo is particularly important in computer science because
of its decidability and connections with other formalisms like modal,
temporal or description logics or applications in XML or ontology
reasoning. The satisfiability of \FOTwo was proved to be decidable in
\cite{Scott1962,Mortimer75} and \NEXPTIME-complete in
\cite{GrKoVa97}.

All decidable fragments of first-order logic have limited expressive
power and a lot of effort is being put to extend them beyond
first-order logic while preserving decidability.  Many extensions of
\FOTwo, in particular with transitive closure or least fixed-point
operators, quickly lead to undecidability
\cite{GraedelOttoRosen99,IRRSY-CSL04}. Extensions that go beyond first
order logic, but their (finite) satisfiability problem remains
decidable, include \FOTwo over restricted classes of structures where
one \cite{KieronskiOttoLics05} or two relation symbols
\cite{KieronskiT09} are interpreted as equivalence relations (but
there are no other binary symbols); where one \cite{Otto2001} or two
relations are interpreted as linear orders~\cite{SchwentickZ10}; where
two relations are interpreted as successors of two linear
orders~\cite{Manuel10, Figueira2012, CharatonikWitkowski-lics13};
where one relation is interpreted as linear order, one as its
successor and another one as equivalence \cite{BojanczykDMSS11}; where
one relation is transitive~\cite{szwastT13}; where an equivalence
closure can be applied to two binary
predicates~\cite{KieroMPT-lics12}; where deterministic transitive
closure can be applied to one binary
relation~\cite{CharatonikKM14}. It is known that the finite
satisfiability problem is undecidable for \FOTwo with two transitive
relations~\cite{Kieronski-csl12}, with three equivalence
relations~\cite{KieronskiOttoLics05}, with one transitive and one
equivalence relation~\cite{KieronskiT09}, with three linear
orders~\cite{Kieronski-csl11}, with two linear orders and their two
corresponding successors~\cite{Manuel10}. A summary of complexity
results for extensions of \FOTwo with binary predicates being the
order relations (in the absence of other binary relations) can be
found in~\cite{ManuelZ13}.

When studying extensions of \FOTwo it is enough to consider relational
signatures with symbols of arity at most $2$~\cite{GrKoVa97}. Some of
the above mentioned decidability results, \eg~\cite{BojanczykDMSS11,
  CharatonikKM14, Figueira2012, Manuel10, SchwentickZ10}, hold under
the assumption that there are no other binary predicates in the
signature; some, like~\cite{CharatonikWitkowski-lics13,
  KieroMPT-lics12, KieronskiOttoLics05, KieronskiT09, Otto2001,
  szwastT13} are valid in the general setting. Similarly, for
undecidability results additional binary symbols are required \eg
in~\cite{GraedelOttoRosen99,IRRSY-CSL04,KieronskiT09}, but not
required in~\cite{Kieronski-csl11, Kieronski-csl12,
  KieronskiOttoLics05, Manuel10}.

The two-variable fragment with counting quantifiers (\CTwo) extends
\FOTwo by allowing counting quantifiers of the form $\exists^{<k}$,
$\exists^{\leq k}$, $\exists^{=k}$, $\exists^{\geq k}$ and
$\exists^{>k}$, for all natural numbers $k$.  The two problems of
satisfiability and finite satisfiability for \CTwo (which are two
different problems as \CTwo does not have a finite model property)
were proved to be decidable in~\cite{GradelORLics97}. Another
decidability proof together with a~\NEXPTIME-completeness result under
unary encoding of numbers in counting quantifiers can be found in
\cite{PacholskiSTLics97}. Pratt-Hartmann in~\cite{IPH05} established
\NEXPTIME-completeness of both satisfiability and finite
satisfiability under binary encoding of numbers in counting
quantifiers. All these algorithms are quite sophisticated, a
significant simplification can be found in~\cite{IPH10}.  There are
not many known decidable extensions of \CTwo. In
\cite{CharatonikWitkowski-lics13} it is shown that finite
satisfiability for \CTwo interpreted over structures where two binary
relations are interpreted as forests of finite trees (which subsumes
the case of two successor relations on two linear orders) is
\NEXPTIME-complete. \cite{Pratt-lics14} shows that the satisfiability
and finite satisfiability problems for \CTwo with one equivalence
relation are both \NEXPTIME-complete. All extensions of \CTwo
mentioned above allow arbitrary number of other binary
relations.\enlargethispage{\baselineskip}

In this paper we study extensions of \CTwo with linear orders.  We
show that the finite satisfiability problem for \CTwo with two linear
orders, in the presence of other binary predicate symbols, is
undecidable.  For \CTwo with one linear order, even if the successor
of this linear order is present, in the absence of other binary
predicate symbols, it is decidable and \NEXPTIME-complete. A
surprising result is that when we add one more binary predicate
symbol, the complexity of the problem explodes: the finite
satisfiability problem for \CTwo with one linear order and in the
presence of one more binary predicate symbol is equivalent, under
elementary reductions, to the emptiness problem for multicounter
automata. Thus it is decidable, but as complex as the
reachability problem for vector addition systems.

Multicounter automata (MCA) are a~very simple formalism equivalent to
Petri Nets and vector addition systems (VAS)~\cite{Nash73}, which are
used \eg to describe distributed, concurrent systems and
chemical/biological processes.  One of the main reasoning tasks for
VAS is to determine reachability of a given vector.  It is known that
this problem is decidable~\cite{Kosaraju82,Mayr84,Leroux09} and
\EXPSPACE-hard~\cite{Lipton76}, but precise complexity is not known,
and after over 40 years of research it is even not known if the
problem is elementary. We present a reduction from the emptiness
problem of MCA (which is equivalent to the reachability for VAS) to
finite satisfiability of \CTwo with one linear order
and one more binary predicate symbol.  Although we show that \CTwo
with one linear order and its successor, in the presence of
arbitrary number of binary predicate symbols is decidable, it is very
unlikely that it has an elementary decision algorithm since existence
of such an algorithm implies existence of an elementary algorithm for
VAS reachability.

The current paper is a~revised and extended version
of~\cite{CharatonikWitkowski-csl15}.

\section{Preliminaries} 
We will consider finite satisfiability problems for the two-variable
logic with counting (\CTwo for short) over finite structures, where
some distinguished binary symbols are interpreted as linear orders or
successors of linear orders. We will be interested in largest
antireflexive relations $\Lin$ contained in linear orders $\leq$; for
a linear order $\leq$ we write $a \Lin b$ iff $a\leq b$ and $a\neq
b$. In the rest of the paper by a~linear order we mean the relation
$\Lin$ rather than~$\leq$.  We use symbols $\Lin$, $\fLin$, $\sLin$ to
denote linear orders and $\Succ$, $\fSucc$, $\sSucc$ to denote their
respective successors. Given a finite signature $\Sigma$ we write
$\ClassO$ to denote the class of finite structures over $\Sigma$,
where $\Lin$ is interpreted as a~linear order and $\Succ(x,y)$ means
that $y$ is a~successor of $x$ in this order. We also adopt a similar
notation for other classes of structures. Logics we consider will be
denoted by $\CTwoGen{\#_u}{\#_b}{I}$, where $\#_u, \#_b\in
\Nat\cup\{*\}$ denote the number of unary \resp binary symbols allowed
in formulas, and $I$ is the signature of distinguished binary
symbols. We use * to denote that $\#_u$ \resp $\#_b$ are unbounded and
we do not count the size of $I$ in $\#_b$.

Specifically, we will be interested in the following logics:
\CTwoOneLin, \CTwoOneLinS and \CTwoLinUndecid. By \CTwoOneLin we mean
the logic \CTwo over a~signature that contains an arbitrary number of
unary and binary predicates, and $\Lin$ is interpreted as a linear
order. The definition of \CTwoOneLinS is similar, with the exception
that $\Succ$ is interpreted as the successor of $\Lin$. The logic
\CTwoLinUndecid allows an arbitrary number of unary and at most $4$
binary symbols, two of them are interpreted as linear orders. Notice
that we do not allow constant symbols in signatures, but this does not
cause loss of generality since constants can be simulated by unary
predicates and counting quantifiers.

\section{Two linear orders}\label{sec:twoorders}
We start with the observation that the induced successor relation of
a linear order can be expressed in \CTwoOneLin. More precisely, let
$s$ be a non-distinguished binary predicate. The following lemma says
that $s$ can be defined to mean the successor of $\Lin$ in
$\CTwoLinOne$.  Intuitively, it is enough to state that $s$ is
a~subrelation of $\Lin$ such that each element (with the exception of the
least and the greatest one) has exactly one $s$-successor and exactly
one $s$-predecessor.

\begin{lemma}\label{lem:succ}
  There exists a formula $\varphi_{s}$ of $\CTwoLinOne$ such that for
  every finite structure~$\str M$, we have $\str M \models
  \varphi_{s}$ if and only if $s^{\str M}$ is the induced successor
  relation of $\Lin^{\str M}$. 
\end{lemma}

\begin{proof}
  Define $\varphi_{s}$ as conjunction of the following three formulas.
\begin{gather}
  \forall{x}\forall{y}. s(x,y) \rightarrow x \Lin y \label{succ:first} \\
  \forall{x}.\left(\forall{y}. y < x \vee y = x\right) \vee \exists^{=1}{y}. s(x,y)\label{succ:second}\\
  \forall{y}.\left(\forall{x}. y < x \vee y = x\right) \vee
  \exists^{=1}{x}. s(x,y)\label{succ:third}
\end{gather}
Conjunct~(\ref{succ:first}) of $\varphi_{s}$ says that $s^{\str M}$ is
a~subrelation of $\Lin^{\str M}$. Conjuncts~(\ref{succ:second})
and~(\ref{succ:third}) state that every non-least (respectively,
non-greatest) element \wrt $\Lin^{\str M}$ has precisely one
$s$-successor (respectively, $s$-predecessor).

Let $\str M$ be a finite model of $\varphi_{s}$ and let $e_1$ be the
least \wrt $\Lin^{\str M}$ element of $\str M$. By a simple inductive
argument we can construct a~sequence $e_1,\ldots e_k$ of elements of
$\str M$ such that $s^{\str M}(e_i,e_{i+1})$ holds for all
$i\in\{1,\ldots, k-1\}$ and $e_k$ is the greatest element in $\str M$.
By conjuncts~(\ref{succ:second}) and~(\ref{succ:third}) no element
occurs more than once in this sequence.

Observe that if the sequence $e_1,\ldots e_k$ contains all elements of
$\str M$ then $s^{\str M}$ is the induced successor relation of
$\Lin^{\str M}$. Now, seeking for contradiction, assume that $e$ is an
element of $\M$ that does not appear in the sequence. Then, again by
induction, we can construct another finite sequence that contains $e$,
is disjoint with $e_1,\ldots e_k$ (so it contains neither the least
not the greatest element in~$\str M$), and where every element has an
$s$-successor and an $s$-predecessor. It follows that the second
sequence is an $s$-cycle, which contradicts
conjunct~(\ref{succ:first}).

For the other direction, if $\str M$ is such that $s^{\str M}$ is the
successor relation of linear order~$\Lin^{\str M}$, then it can be
checked by inspection that $\str M\models \varphi_{s}$.
\end{proof}

\begin{corollary}\ Finite satisfiability of \CTwoOneLinS is reducible
  in linear time to finite satisfiability of \CTwoOneLin.  Finite
  satisfiability of \CTwoLinUndecidS is reducible in linear time to
  finite satisfiability of \CTwoLinUndecid.
\end{corollary}

Since \FOTwoLinUndecidS, \ie the two-variable logic with two linear
orders and their corresponding successors, is
undecidable~\cite{Manuel10}, we have the following conclusion.
\begin{corollary}
  Finite satisfiability problem of \CTwoLinUndecidA is undecidable.
  This remains true even for \CTwoLinUndecid.
\end{corollary}

\begin{remark}\label{rem:succ}
  Observe that the proof of Lemma~\ref{lem:succ} works also in
  a~setting where the symbol $\Lin$ is interpreted as an arbitrary
  acyclic relation (a relation is acyclic if its diagram is an acyclic
  graph). Actually, assuming that $\Lin^{\str M}$ is an acyclic
  relation, the formula $\varphi_s$ forces it to be the linear order
  induced by $s^{\str M}$. 
\end{remark}

\section{\texorpdfstring{\CTwoLinSucc}{C2[...]} is \texorpdfstring{\NEXPTIME}{NEXPTIME}-complete}
\label{sec:linsucc}

We will show that finite satisfiability problem for \CTwoLinSucc is
\NEXPTIME-complete. Since the lower bound follows from the complexity
of \FOTwo with only unary predicates~\cite[Theorem 11]{EtessamiVW02},
we will concentrate on proving the upper bound. The proof presented
here is inspired by a corresponding result~\cite{CKM13} on \FOTwo on
finite trees: we bring the input formula to a~normal form and then
check its local and global consistency, using a~notion of types that
is similar to the one introduced in~\cite{CKM13} .

We will be interested in \CTwoLinSucc formulas $\varphi$ in normal
form
\begin{equation}
  \label{nf1}
\varphi = \forall{x}\forall{y}.\chi(x,y)\wedge \bigwedge_{h=1}^m
\forall{x}\exists^{\lessdot_h C_h}{y}.\chi_h(x,y),
\end{equation}
where $\chi, \chi_1,\ldots,\chi_{m}$ are quantifier-free formulas.
Here symbols $\lessdot_h$, for $h=\{1,\ldots,m\}$, denote either
$\leq$ or $\geq$, and $C_1,\ldots, C_m$ are positive integers encoded
in binary. We refer to the number $c=\max\{C_h\mid h\in \{1,\ldots,m
\}\}$ as the \emph{height} of the formula $\varphi$.
It is well known~\cite[Theorem 2.2]{GradelO99} that by adding
additional unary predicates each $C^2$ formula $\varphi$ can be
transformed in polynomial time to a formula in normal form that is
equisatisfiable with~$\varphi$ over models of cardinality at least
$c$.

Observe that \CTwoLinSucc may be seen as a fragment of the weak
monadic second-order logic with one successor WS1S, where unary
relations are simulated by second-order existential quantifiers and
counting quantifiers by first-order ones (e.g., a formula of the form
$\exists^{\leq k}x.\chi(x)$ can be replaced by an equivalent formula
with $k+1$ universal quantifiers). However, this view leads to
formulas with three alternations of quantifiers that can be checked
for satisfiability in $4\EXPTIME$, which is not a~desired complexity
bound.

Because an element of a model of $\varphi$ may require up to $c$
witnesses for satisfaction, we will be interested in multisets
counting these witnesses.  Let $\Nat_{c} = \{n\in\Nat\mid n\leq c\}
\cup \{\infty\}$. For $k,k'\in \Nat_{c}$ define $\mathrm{cut}_c(k) =
k$ if $k \leq c$ and $\mathrm{cut}_c(k) = \infty$ if $k>c$.  Define
$k\oplus_c{k'} = \mathrm{cut}_c(k+k')$. A $c$-multiset of elements
from a given set $A$ is any function $f:A\to \Nat_c$. For a given
element $a$ in $A$, the singleton $\{a\}$ is the multiset defined by
$\{a\}(x)=1$ if $x=a$ and $\{a\}(x)=0$ for $x\neq a$.  The union of
two multisets $f$ and $g$ is a function denoted $f\cup g$ such that
$(f\cup g)(x)=f(x)\oplus_cg(x)$. We say that $f$ is a
\emph{submultiset} of $g$, written $f \subseteq g$, if $f(a) \leq
g(a)$ for all $a\in A$. The empty multiset, denoted $\emptyset$, is
the constant function equal $0$ for all arguments. The following fact
is obvious.

\begin{fact}\label{fact:chainSize}
Every ascending (\resp descending) \wrt $\subseteq$ chain of $c$-multisets
consists of at most $|A|*(c+2)$ distinct elements.
\end{fact}

Let us call maximal consistent formulas specifying the relative
position of a pair of elements in a structure in $\ClassO$ \emph{order
  formulas}.  There are five possible order formulas: $x {=} y\wedge
\neg\Succ(y,x)\wedge\neg\Succ(x,y)\wedge y{\not<}x \wedge x{\not <}y$,
\mbox{$x{\neq}y \wedge\Succ(y,x)\wedge\neg\Succ(x,y)\wedge y{<}x
  \wedge x{\not <}y$}, $x{\neq}y
\wedge\Succ(x,y)\wedge\neg\Succ(y,x)\wedge x{<}y \wedge y{\not <}x$,
$x{\neq}y \wedge\neg\Succ(x,y)\wedge\neg\Succ(y,x)\wedge x{<}y \wedge
y{\not <}x$, and \mbox{$x{\neq}y
\wedge\neg\Succ(x,y)\wedge\neg\Succ(y,x)\wedge y{<}x \wedge
x{\not <}y$}.  They are denoted, respectively, as: $\teq$, $\tprec$,
$\tsucc$, $\tless$, $\tgreat$. Let $\Theta$ be the set of these five
formulas.

A \emph{1-type} over the signature $\Sigma$ is a~maximal consistent
conjunction of atomic and negated atomic formulas over $\Sigma$
involving only the variable $x$.  The set of all 1-types over $\Sigma$
will be denoted $\Pi(\Sigma)$.  The family of all $c$-multisets of 1-types
over the signature $\Sigma$ is denoted~$\Nat_c^{\Pi(\Sigma)}$.

To be able to check local consistency of a~formula (see
Definition~\ref{def:compatible} below) we introduce the notion of
a~\emph{full type}. Intuitively, a~full type of an element $e$ in a~structure
contains enough information to check that the formula is (locally)
true in $e$.  The information stored in the full type of $e$ tells us
about 1-types of all other elements $e'$ in the structure, divided
into five multisets depending on relative positions of $e$ and $e'$.

\begin{definition}[Full type over $\Sigma$ \wrt $c$]
  A \emph{full type} over $\Sigma$ \wrt $c$ is a function
  $\sigma:\Theta \to \Nat_c^{\Pi(\Sigma)}$, such that
  $\sigma(\tprec)$ and $\sigma(\tsucc)$ are singletons or
  empty, and $\sigma(\teq)$ is a singleton. 
\end{definition}

\begin{definition}[Full type in $\str A$ \wrt $c$]
  Let $\str A$ be a structure over a signature $\Sigma$ and let $a$
  be an element of $\str A$.  The \emph{full type} of $a$ in $\str A$,
  denoted $\mathrm{ft}^{\str A}(a)$ is the function $\sigma:\Theta\to
  \Nat_c^{\Pi(\Sigma)}$ such that
  \begin{itemize}
  \item $\sigma(\teq)$ is the singleton of the 1-type of $a$ in $\str A$,
  \item $\sigma(\tprec)$ is the singleton of the 1-type of the
    predecessor of $a$ (if $a$ has a predecessor) or empty multiset
    (if $a$ has no predecessor),
  \item $\sigma(\tsucc)$ is the singleton of the 1-type of the
    successor of $a$ (if $a$ has a successor) or empty multiset
    (if $a$ has no successor),
  \item $\sigma(\tless)$ is the $c$-multiset of 1-types of elements
    strictly smaller than $a$ in $\str A$, excluding the predecessor
    (if it exists), and
  \item $\sigma(\tgreat)$ is the $c$-multiset of 1-types of elements
    strictly greater than $a$ in $\str A$, excluding the successor (if
    it exists).
  \end{itemize}
  A structure $\str A$ is said to \emph{realise} a full type $\sigma$ if
  $\mathrm{ft}^{\str A}(a)=\sigma$ for some $a\in \str A$. In the
  following we identify a full type $\sigma$, which is a
  function, with the tuple $\tuple{ \sigma(\tprec), \sigma(\teq),
    \sigma(\tsucc), \sigma(\tless), \sigma(\tgreat) }$.
\end{definition}

Let $\sigma$ be a full type \wrt $c$ such that $\sigma(\teq)=\{\pi\}$ and let
$\forall{x}\exists^{\lessdot_h C_h}{y}.\chi_h(x,y)$ be a conjunct
in~$\varphi$. The following five functions are used to count witnesses
\wrt this conjunct for elements of full type $\sigma$.
  \begin{eqnarray*}
    W_{=}^{\chi_h}(\sigma)&=&
    \begin{cases}
      1 & \text{if }  \pi(x)\models \chi_h(x,x) \\
      0 & \text{otherwise}
    \end{cases}\\
    W_{\Pred}^{\chi_h}(\sigma)&=&
    \begin{cases}
      1 & \text{if }  \sigma(\tprec)=\{\pi'\}\text{ and }\pi(x)\wedge\pi'(y)\wedge \tprec(x,y)\models \chi_h(x,y) \\
      0 & \text{otherwise}
    \end{cases}\\
    W_{\Succ}^{\chi_h}(\sigma)&=&
    \begin{cases}
      1 & \text{if } \sigma(\tsucc)=\{\pi'\}\text{ and }\pi(x)\wedge\pi'(y)\wedge \tsucc(x,y)\models \chi_h(x,y) \\
      0 & \text{otherwise}
    \end{cases}
  \end{eqnarray*}
  \begin{eqnarray*}
    W_{<}^{\chi_h}(\sigma)&=&
\mathrm{cut}_c\Big(\sum_{\pi':\pi(x)\wedge\pi'(y) \wedge \tless(y,x)\models \chi_h(x,y)}
(\sigma(\tless))(\pi')\Big)\\
    W_{>}^{\chi_h}(\sigma)&=&
\mathrm{cut}_c\Big(\sum_{\pi':\pi(x)\wedge\pi'(y) \wedge \tgreat(y,x)\models \chi_h(x,y)}
(\sigma(\tgreat))(\pi')\Big)\\
  \end{eqnarray*}
  Note that in the definition above $(\sigma(\tgreat))(\pi')$ is
  simply the number of occurrences of the 1-type $\pi'$ in the
  multiset $\sigma(\tgreat)$.

  The following definition and lemma formalise local consistency of
  a~formula.
  \begin{definition}[Compatible full types] \label{def:compatible} Let
    $\sigma$ be a full type \wrt $c$ such that $\sigma(\teq)=\{\pi\}$
    and let $\varphi$ be a formula of height $c$ in normal
    form~(\ref{nf1}).  We say that $\sigma$ is \emph{compatible}
    with~$\varphi$ if the following conditions are satisfied.
  \begin{itemize}
  \item $\pi(x)\models \chi(x,x)$,
  \item $\pi(x)\wedge\pi'(y)\wedge\theta(x,y)\models \chi$ for all
    $\theta\in\{\tprec,\tsucc,\tless,\tgreat\}$ and all
    $\pi'\in\sigma(\theta)$, and
  \item  for each conjunct $\forall{x}\exists^{\lessdot_h C_h}{y}.\chi_h(x,y)$ of $\varphi$ we have 
\[
W_{=}^{\chi_h}(\sigma)+W_{\Pred}^{\chi_h}(\sigma)+W_{\Succ}^{\chi_h}(\sigma)+W_{<}^{\chi_h}(\sigma)+W_{>}^{\chi_h}(\sigma)\lessdot_h C_h
\]
  \end{itemize}  \smallskip
\end{definition}

\noindent It is quite obvious that whenever $\str A\models \varphi$, all full
types realised in $\str A$ are compatible with~$\varphi$.  It is not
difficult to see that the converse is also true, as the following
lemma says.

\begin{lemma}\label{lem:compatible}
  For any ordered structure $\str A$ and any \CTwoLinSucc formula $\varphi$ in
  normal form, if all full types realised in $\str A$ are compatible
  with~$\varphi$ then $\str A\models \varphi$.
\end{lemma}

\begin{proof}
  Take arbitrary two elements of the structure $\str A$. If the two
  elements are equal then the first item of
  Definition~\ref{def:compatible} guarantees that they satisfy the
  conjunct $\chi$ of $\varphi$. If they are different, $\chi$ is
  satisfied by the second item. In any case, the conjunct $\chi$ is
  satisfied. Similarly, take any element $e$ of the structure and
  consider any conjunct $\forall{x}\exists^{\lessdot_h
    C_h}{y}.\chi_h(x,y)$ of $\varphi$. Each element $e'$ such that
  $\str A\models \chi(e,e')$ belongs to exactly one of the five sets:
  the singleton of $e$, the singleton of the predecessor of $e$, the
  singleton of the successor of~$e$, elements smaller than the
  predecessor of $e$ and the elements greater than the successor
  of~$e$. The five functions used in the third item of
  Definition~\ref{def:compatible} correctly (up to $c$) count the
  number of such elements $e'$ in these five sets. Since the constant
  $C_h$ does not exceed $c$, the the condition in the third item
  guarantees that the conjunct is satisfied.
\end{proof}

To be able to check global consistency of a~formula (see
Lemma~\ref{lem:graphreduction} below) we introduce the notion of the
graph $\sigmacgraph$. It is the graph $\tuple{V,E}$ where the set $V$
of nodes is the set of full types over $\Sigma$ \wrt $c$ and the set
$E$ of edges is defined as follows.
\begin{eqnarray*}
 {\tuple{\tuple{\Pi_{\Pred}, \{\pi\}, \{\pi_{\Succ}\},\Pi_<, \Pi_>},\tuple{\{\pi\}, \{\pi_{\Succ}\},  \Pi'_{\Succ},\Pi'_<, \Pi'_>}}\in E} 
 & \mathrm{iff}& \Pi'_<=\Pi_<\cup \Pi_{\Pred} \mathrm{~~and~~}\\ && \Pi_>=\Pi'_>\cup \Pi'_{\Succ}
\end{eqnarray*}
Intuitively, a path in this graph describes a possible evolution of
full types in a structure. Each edge in the graph describes the change
in a full type when we move from an element in the structure to its
successor: (the singleton containing the 1-type of) the successor is
moved to the position of current element; the current element is moved to
the predecessor position; the predecessor is added to (the multiset of
the 1-types of the) strictly smaller elements; and finally the multiset
of strictly greater elements is split into the new multiset of strictly
greater elements and the new successor element.

We define the graph $\sigmacphigraph$ as the subgraph of
$\sigmacgraph$ consisting of nodes compatible with~$\varphi$. The
nodes of the form $\tuple{\emptyset,\ldots, \ldots,\emptyset, \ldots}$
are called \emph{source} nodes; the nodes of the form $\tuple{\ldots,
  \ldots, \emptyset, \ldots,\emptyset}$ are called \emph{target}
nodes.  Intuitively, a source node corresponds to a full type of the
least element in some model of~$\varphi$ while a target node
corresponds to the greatest element in some model.

\begin{lemma}\label{lem:graphreduction}
  Let $\varphi$ be a \CTwoLinSucc formula of height $c$ in normal
  form, over signature~$\Sigma$. Then~$\varphi$ is finitely
  satisfiable if and only if there exists a path from a source node to
  a target node in the graph $\sigmacphigraph$.
\end{lemma}
\begin{proof}
  First, assume that $\varphi$ is finitely satisfiable, and let $\str
  A$ be a model of $\varphi$. Let $a_1,\ldots,a_k$ be all elements of
  $\str A$ ordered \wrt $\Lin^{\str A}$. Then $\mathrm{ft}^{\str
    A}(a_1),\ldots,\mathrm{ft}^{\str A}(a_k)$ is a path from a source
  node to a target node in $\sigmacphigraph$.

  Second, assume that there exists a path $\sigma_1,\ldots,\sigma_k$
  from a source node to a target node in $\sigmacphigraph$.
  We will construct a structure $\str A$ over elements $a_1,\ldots
  a_k$ such that $\str A\models \varphi$. Let the universe of $\str A$
  consist of $k$ distinct elements $a_1,\ldots a_k$. Define unary
  predicates in $\str A$ in such a~way that for all $i\in
  \{1,\ldots,k\}$ the 1-type of $a_i$ coincides with
  $\sigma_i(\teq)$. Then define the relation $\Succ^{\str A}$ as
  $\{\tuple{a_i,a_{i+1}}\mid 1\leq i\leq k-1\}$ and $\Lin^{\str A}$ as
  the transitive closure of $\Succ^{\str A}$.

  The definition of $\Succ^{\str A}$ guarantees that for
  $\theta\in\{\tprec,\tsucc\}$ and for all $i$ we have
  $\sigma_i(\theta)=\mathrm{ft}^{\str A}(a_{i})(\theta)$. A~simple
  inductive proof shows that $\sigma_i$ coincides with
  $\mathrm{ft}^{\str A}(a_{i})$ on $\tless$ and $\tgreat$, too. Since
  $\sigma_1$ is a~source node, $\sigma_1(\tless)$ is the empty
  multiset, which coincides with the value of $\mathrm{ft}^{\str
    A}(a_{1})$ on $\tless$. In the inductive step, assuming that
  $\sigma_i(\tless)$ and $\mathrm{ft}^{\str A}(a_{1})(\tless)$
  coincide, we show the same for $\sigma_{i+1}$ and $\mathrm{ft}^{\str
    A}(a_{i+1})$: both the edge relation $E$ and the definition of
  full types require that the values of $\sigma_{i+1}$ and
  $\mathrm{ft}^{\str A}(a_{i+1})$ on $\tless$ are the $c$-multiset
  union of the value of $\sigma_i$ on $\tless$ with the singleton of
  the 1-type of the predecessor of $a_i$. The case of the value on
  $\tgreat$ is symmetric: the induction starts in $k$ with empty
  multisets and goes down to $1$ adding in each step respective
  successors.  This shows that $\mathrm{ft}^{\str A}(a_{i}) =
  \sigma_i$ for all $i\in \{1,\ldots, k\}$. Therefore all full types
  realised in $\str A$ are compatible with $\varphi$ and by
  Lemma~\ref{lem:compatible} we have $\str A\models \varphi$.
\end{proof}

Lemma~\ref{lem:graphreduction} leads us directly to the main theorem
of this section. To check satisfiability of a formula in \CTwoLinSucc
it is enough to guess an appropriate path in
$\sigmacphigraph$. Moreover, it is enough to use only
exponentially many different full types in the guessed path.
\begin{theorem}\label{thm:linsucc}
   The finite satisfiability problem for \CTwoLinSucc is
\NEXPTIME-complete.
\end{theorem}
\begin{proof}
  The lower bound follows from the complexity of the finite
  satisfiability problem for \FOTwo with only unary predicates. For
  the upper bound, an algorithm for deciding finite satisfiability of
  \CTwoLinSucc works as follows. It takes a \CTwoLinSucc formula
  $\psi$ and converts it to a normal form $\varphi$ (in polynomial
  time).  It also checks (by nondeterministic guessing) if there
  exists a~model of $\psi$ of cardinality at most $c$, where $c$ is
  the height of $\varphi$.  Then the algorithm guesses a path from a
  source node to a target node in $\sigmacphigraph$ where $\Sigma$ is
  the signature of $\varphi$. This requires in particular
  verification of the fact that all nodes are compatible with
  $\varphi$.

  All this can be accomplished in time polynomial in the size of the
  graph. This size is potentially doubly exponential in $|\varphi|$:
  the number of all 1-types over $\Sigma$ is exponential in
  $|\varphi|$, so the number of sets of 1-types, and, in consequence,
  the number of full types, is doubly exponential. The potential
  2\NEXPTIME\ complexity of the algorithm can be lowered to
  \NEXPTIME\ using the observation that the $\tless$ and $\tgreat$
  components of full types behave in a~monotone way along any path
  connecting any source node with any target node. The $\tless$
  component may only increase and $\tgreat$ only decrease along any
  such path. Thus multisets on a path form an ascending \resp
  descending chain, and by Fact~\ref{fact:chainSize} there are only
  $O(|\Pi(\Sigma)|*c)$ such multisets occurring along the
  path. Therefore it is enough to guess only exponentially many (in
  $|\varphi|$) different full types.
\end{proof}

\section{Hardness of \texorpdfstring{\CTwoLinOne}{C2[...]}}\label{sec:hard}
In this section we show that the finite satisfiability problem for
\CTwoLinOne, with a binary relation $s$, is at least as hard as
non-emptiness of multicounter automata.  Similar reductions from
emptiness of multicounter automata in this context can be found
e.g. in~\cite{BojanczykDMSS11} and~\cite{ManuelZ13}.  Here, for a
given multicounter automaton $M$, we first construct a \CTwoSBin
formula~$\varphi_{M}$ which has a finite model if and only if $M$ is
non-empty. Then, using the idea from Section~\ref{sec:twoorders}, we
argue that the reduction can be modified to work for \CTwoLinOne.

We adopt a notion of {multicounter automata} (MCA for short)
similar to one in~\cite{BojanczykDMSS11} or~\cite{ManuelZ13}, but with
empty input alphabet and simplified counter manipulation. The
exposition below closely follows the one from the long version
of~\cite{ManuelZ13}. Intuitively, an MCA is a finite state automaton
without input but equipped with a finite set of counters which can be
incremented and decremented, but not tested for zero.  More formally,
a multicounter automaton $M$ is a tuple $\tuple{Q, C, R, \delta, q_I,
  F}$, where the set $Q$ of states, the initial state $q_I \in Q$ and
the set $F \subseteq Q$ of final states are as in usual finite state
automata, $C$ is a finite set (the \emph{counters}) and $R$ is a
subset of $C$. The transition relation $\delta$ is a subset of
\[
Q \times \{inc(c), dec(c), skip\mid c \in C\} \times Q.
\] 
An MCA is called \emph{reduced} if it does not have skip transitions
and $R=C$ (in this case we just omit the $R$ component of tuple $M$).

A \emph{configuration} of a multicounter automaton $M$ is a pair
$\tuple{p, \vec{n}}$ where $p$ is a state and $\vec{n} \in \Nat^C$
gives a value $\vec{n}(c)$ for each counter $c$ in $C$. Transitions
with $inc(c)$ and $skip$ can always be applied, whereas transitions
with $dec(c)$ can only be applied to configurations with $\vec{n}(c) >
0$. Applying a transition $\tuple{p, inc(c), q}$ to a configuration
$\tuple{p, \vec{n}}$ yields a configuration $\tuple{q, \vec{n}_0 }$
where $\vec{n}_0$ is obtained from $\vec{n}$ by incrementing its
$c$-th component and keeping values of all other components
unchanged. Analogously, applying (an applicable) transition $\tuple{p,
  dec(c), q}$ to a configuration $\tuple{p, \vec{n}}$ yields a
configuration $\tuple{q, \vec{n}_0}$ where $\vec{n}_0$ is obtained
from $\vec{n}$ by decrementing its $c$-th component. Transitions with
$skip$ do not change value of any counter in $C$. A \emph{run} is an
interleaving sequence of configurations and transitions
$\conf_1,\trans_1,\ldots,\trans_{k-1},\conf_k$ such that $\trans_i$
applied to $\conf_i$ gives $\conf_{i+1}$, for $1\leq i<k$. A run is
\emph{accepting}, if it starts in configuration $\tuple{q_I ,\vec{0}}$
and ends in some configuration $\tuple{q_F , \vec{n}_F}$ with $q_F \in
F$ and $\vec{n}_F(c) = 0$ for every $c\in R$ (note that for counters
$c\in {C\setminus R}$ the value $\vec{n}_F(c)$ may be arbitrary). The
emptiness problem for multicounter automata is the question whether a
given automaton $M$ has an accepting run. It is well known that this
problem (for both MCA and reduced MCA) is decidable, as it is
polynomial-time equivalent to the reachability problem in Vector Addition
Systems/Petri Nets\cite{Kosaraju82, Mayr84}.

\begin{definition}\label{def:hardFormula}
  Let $M = \tuple{Q, C, \delta, q_I, F}$ be a {reduced} MCA. Let
  $\Sigma = \{q\mid q\in Q\} \cup \{inc_c, dec_c\mid c\in C\} \cup
  \{\min,\max, 
  {\Lin}, \Succ, s\}$ where predicates $q, inc_c, dec_c, \min$ and
  $\max$ are unary and $\Lin, \Succ$ and $s$ are binary. Define
  $\varphi_M$ as the conjunction of the following $\Sigma$-formulas.
\begin{enumerate}[widest=10]
\item $\exists^{=1}{x}.\min(x) \wedge
  \exists^{=1}{x}.\max(x)$\label{def:hardFormula:1}
\item $\forall{x}\forall{y}. \left(\min(x) \rightarrow
  \left({x}\Lin{y} \vee x = y \right)\right) \wedge \left(\max(x)
  \rightarrow \left(y\Lin{x} \vee y = x
  \right)\right)$\label{def:hardFormula:2}
\item $\forall{x}.\left(\bigvee_{q\in Q} q(x)\right) \wedge
  \bigwedge_{q\in Q} \left(q(x)\rightarrow \bigwedge_{q'\in
    Q\setminus\{q\}} \neg q'(x)\right)$\label{def:hardFormula:3}
\item $\forall{x}. \left(\min(x) \rightarrow q_I(x)\right) \wedge
  \left(\max(x) \rightarrow \bigvee_{q_F\in F}
  q_F(x)\right)$\label{def:hardFormula:4}
\item $\!\!\!\begin{array}[t]{ll} \lefteqn{\forall{x}\forall{y}. {\Succ(x,y)}
  \rightarrow }\mbox{~~~}\\ &\bigvee_{\tuple{q,inc(c),q'}\in \delta}
  \left(q(x) \wedge inc_c(x) \wedge q'(y)\right)\; \vee\;
  \bigvee_{\tuple{q,dec(c),q'}\in \delta} \left(q(x) \wedge dec_c(x)
  \wedge q'(y)\right)
  \end{array}$
\label{def:hardFormula:5}
\item $\forall{x}.  \left(\neg\max(x)\right)\rightarrow \bigvee_{c\in C}\left(inc_c(x)\vee dec_c(x)\right)$\label{def:hardFormula:6}
\item $\forall{x}.\bigwedge_{c\in C} \left(inc_c(x)\rightarrow \neg dec_c(x)
    \wedge \bigwedge_{c'\in C\setminus\{c\}} \left(\neg dec_{c'}(x)\wedge\neg inc_{c'}(x)\right)\right)$\label{def:hardFormula:7}
\item $\forall{x}.\bigwedge_{c\in C} \left(dec_c(x)\rightarrow \neg inc_c(x)
    \wedge \bigwedge_{c'\in C\setminus\{c\}} \left(\neg inc_{c'}(x)\wedge\neg dec_{c'}(x)\right)\right)$\label{def:hardFormula:8}
\item $\forall{x}. \max(x)\rightarrow \bigwedge_{c\in C}\left(\neg
    inc_c(x)\wedge \neg dec_c(x)\right)$\label{def:hardFormula:9}
%\end{enumerate}\begin{enumerate}[label={\cW0}(\arabic*),start=10]
\item $\forall{x}\forall{y}. s(x,y) \rightarrow \bigvee_{c\in C}\left(inc_c(x) \wedge dec_c(y)\right)$\label{def:hardFormula:10}

\item $\forall{x}\forall{y}. s(x,y) \rightarrow {x}\Lin{y}$\label{def:hardFormula:11}

\item $\forall{x}.\left(\max(x) \vee \exists^{=1}{y}.\left(s(x,y) \vee s(y,x)\right)\right)$\label{def:hardFormula:12}

\end{enumerate}
\end{definition}

We will interpret $\varphi_M$ as a \CTwoSBin formula. Models of
$\varphi_M$ encode accepting runs of MCA $M$. The first two conjuncts
of $\varphi_M$ define the meaning of the auxiliary predicates $\min$
and $\max$; they hold for the least (\resp the greatest) element of a
model. Each element of the model corresponds to precisely one state
$q\in Q$, as specified by Conjunct~\ref{def:hardFormula:3}. Thus the
model is just a sequence of states. The first of them must be the
starting state $q_I$ and the last must be a final state $q_F\in F$, as
defined by Conjunct~\ref{def:hardFormula:4}. Every two consecutive
elements of the model form a transition. A state in which the
transition is fired is marked by predicate of the form $inc_c$ or
$dec_c$ denoting a counter to increment or decrement; this is
specified by Conjunct~\ref{def:hardFormula:5}. Every state, with the
exception of the last one, must be labelled by precisely one predicate
of the form $inc_c$ or $dec_c$, as expressed by
Conjuncts~\ref{def:hardFormula:6}--\ref{def:hardFormula:8}. The last
element is not labelled by any of these predicates
(Conjunct~\ref{def:hardFormula:9}), as no transition is fired
there. Since the values of all counters in starting and final state is
$0$ and no counter may fall below $0$, each incrementation of a
counter $c$ must be eventually followed by its decrementation, and
conversely, each decrementation of $c$ must be preceded by its
incrementation. We use the relation $s$ to match these increments and
decrements, as stated in
Conjunct~\ref{def:hardFormula:10}. Conjunct~\ref{def:hardFormula:11}
states that decrementation of a counter indeed follows its
incrementation. Since each state, except the final one, is a starting
state of some transition, it either corresponds to incrementation or
decrementation of some counter. Therefore it emits or accepts
precisely one edge labelled $s$, as stated by
Conjunct~\ref{def:hardFormula:12} of $\varphi_M$.  Formally, we have
the following lemma and a corollary that results from it.
\begin{lemma}\label{lem:phim}
  Let $M = \tuple{Q, C, \delta, q_I, F}$ be a reduced multicounter
  automaton and let $\varphi_M$ be the \CTwoSBin formula constructed in
  Definition~\ref{def:hardFormula}.  Formula $\varphi_M$ is finitely
  satisfiable if and only if $M$ is non-empty.
\end{lemma}
\begin{proof}
  First, assume that $M$ is non-empty. Then $M$ has an accepting run
  of the form $ \conf_1,\trans_1,\ldots,\trans_{k-1},\conf_k$, where
  $\conf_i = \tuple{q_i,\vec{n}_i}$ for $i\in \{1,\ldots,k\}$, and 
  $\trans_i =$ $ \tuple{q_i,inc(c_i),q_{i+1}}$ or $\trans_i =
  \tuple{q_i,dec(c_i),q_{i+1}}$, for $i\in \{1,\ldots,k-1\}$.  We
  will construct a structure $\str A$ on elements $e_1,\ldots,e_k$ and
  show that $\str A$ models $\varphi_M$. Define $(\Lin)^{\str A} =
  \{\tuple{e_i,e_j}\mid i,j\in \{1,\ldots,k\}, i<j\}$ and
  $(\Succ)^{\str A} = \{\tuple{e_i,e_{i+1}}\mid i\in
  \{1,\ldots,k-1\}\}$. Label $e_1$ by predicate $\min$ and
  $e_k$ by $\max$. This leads to satisfaction of
  Conjuncts~\ref{def:hardFormula:1} and~\ref{def:hardFormula:2} of
  $\varphi_M$. Label each element $e_i$ by $q_i$, for $i\in
  \{1,\ldots,k\}$. This leads to satisfaction of
  Conjunct~\ref{def:hardFormula:3}. Conjunct~\ref{def:hardFormula:4} 
  says that $e_1$ should be labelled by $q_I$ and $e_k$ by
  some $q_F$ with $q_F\in F$. Since run $r$ is accepting,
  indeed we have $q_I = q_1$ and $q_k\in F$ and thus
  Conjunct~\ref{def:hardFormula:4} is satisfied by $\str A$. For each
  $i\in \{1,\ldots,k-1\}$ label the element $e_i$ by $inc_{c_i}$ if
  $\trans_i = \tuple{q_i,inc(c_i),q_{i+1}}$ or by $dec_{c_i}$
  if $\trans_i = \tuple{q_i,dec(c_i),q_{i+1}}$. This leads to
  satisfaction of
  Conjuncts~\ref{def:hardFormula:5}--\ref{def:hardFormula:9}.  Finally
  we define relation $s^{\str A}$ in such a way that it
  connects $e_i$ with $e_j$ if $\trans_j$ decrements the same counter
  that is incremented by $\trans_i$ and the value of the counter after
  firing $\trans_j$ is for the first time equal to the value before
  firing $\trans_i$. Formally,
  \begin{eqnarray*}
    s^{\str A}(e_i,e_j)& \mathrm{iff}& i<j,\\
&& \trans_i=\tuple{q_i,inc(c_i), q_{i+1}},\\
&& \trans_j=\tuple{q_j,dec(c_j), q_{j+1}},\\
&& c_i=c_j, \mathrm{ and}\\
&& j=\mathrm{min}\{l\mid l>i\wedge \vec{n}_{l+1}(c_i)=\vec{n}_i(c_i)\}.
  \end{eqnarray*}
  Conjuncts~\ref{def:hardFormula:10} and~\ref{def:hardFormula:11} of
  $\varphi_M$ are satisfied by construction. In the last configuration
  all counters have value $0$, so each incrementation of a counter has
  a later matching decrementation of the same counter. Similarly, in
  the first configuration all counters have value $0$, so each
  decrementation of a counter has an earlier matching incrementation of
  the same counter. Therefore Conjunct~\ref{def:hardFormula:12} is
  also satisfied.

  Second, assume that $\varphi_M$ is finitely satisfiable and let
  $\str A$ be a model of $\varphi_M$. We will show that MCA $M$ is
  non-empty. Assume that elements of $\str A$ sequenced in order
  $(\Lin)^{\str A}$ are $e_1,\ldots,e_k$. By
  Conjunct~\ref{def:hardFormula:3} of $\varphi_M$ there exist a
  sequence of $M$ states $q_1,\ldots,q_k$ such that $q_i^{\str
    A}(e_i)$ hold for every $i\in \{1,\ldots,k\}$. By
  Conjuncts~\ref{def:hardFormula:6}--\ref{def:hardFormula:9} each
  $e_i$ with $i<k$ is labelled by precisely one predicate of the form
  $inc_c^\str{A}(e_i)$ or $dec_c^\str{A}(e_i)$. For $i\in
  \{1,\ldots,k\}$ let
\[\conf_i = \tuple{q_i,\vec{n}_i},\text{ where for every $c\in C$}\]
\[\vec{n}_i(c) = |\{e_j\in \str A\mid\text{ $j<i$ and relation $s^{\str A}(e_j,e_l)$ holds for some $l \geq i$}\}|.\]
For $i\in \{1,\ldots,k-1\}$ let
\[\trans_i = \tuple{q_i,inc(c),q_{i+1}}\text{ provided that $inc_c^\str{A}(e_i)$ holds, or}\]
\[\trans_i = \tuple{q_i,dec(c),q_{i+1}}\text{ provided that $dec_c^\str{A}(e_i)$ holds for some $c\in C$}.\]
The accepting run $r$ of $M$ is
\[\conf_1,\trans_1,\ldots,\trans_{k-1},\conf_k.\]
We will show that $r$ is an accepting run. First we prove that it is a
run of $M$. Assume that for some $i<k$ sequence
$\conf_1,\trans_1,\ldots,\trans_{i-1},\conf_i$ forms a run. We will show
that $\trans_i$ is applicable to configuration $\conf_i$. This is clear
when $\trans_i = \tuple{q_i,inc(c),q_{i+1}}$, for some $c\in
C$. Otherwise $\trans_i = \tuple{q_i,dec(c),q_{i+1}}$. Therefore
$dec_c^{\str A}(e_i)$ holds. Since $i<k$, by
Conjunct~\ref{def:hardFormula:12} of $\varphi_M$, element $e_i$ emits
or accepts precisely one $s$ edge. Because it corresponds
to decrementation of $c$, by Conjunct~\ref{def:hardFormula:10} it must
accept $s$ edge from some element $e_j$. By
Conjunct~\ref{def:hardFormula:11} we have $j<i$. Therefore, by
definition of $\vec{n}_i$ we have $\vec{n}_i(c)>0$. Thus $\trans_i$ is
applicable to $\conf_i$ and
$\conf_1,\trans_1,\ldots,\trans_{i-1},\conf_i,\trans_i,\conf_{i+1}$ indeed
forms a run.

Clearly, $q_1 = q_I$ and $q_k\in F$, by
Conjunct~\ref{def:hardFormula:4}.  Moreover, $\vec{n}_1 =\vec{0}$ and
$\vec{n}_k =\vec{0}$, thus $\conf_1$ and $\conf_k$ are starting and
\resp final configuration. Thus run $r$ is accepting and $M$ is a
non-empty MCA.
\end{proof}

\begin{corollary}
The finite satisfiability problem for \CTwoSBin is at least as hard as
the emptiness problem for multicounter automata.
\end{corollary}

The result above can be strengthened by observing that $\Succ$ symbol
is not necessary. Without this symbol in the signature we may still
encode accepting runs of multicounter automata in \CTwoLinOne by
complicating slightly the definition of the predicate $s$. In the new
encoding its purpose is not only to provide correspondence between
increase and decrease of a counter but also to encode the successor
relationship between elements of a structure. This, however, requires
a slight change in a way we encode accepting runs. In
Definition~\ref{def:hardFormula} a run
$\conf_1,\trans_1,\ldots,\trans_{i-1},\conf_i,\trans_{k-1},\conf_k$ is
represented as a structure with $k$ elements, where the $i$-th element
(for $i<k$) represents both the state of $\conf_i$ and the counter
manipulation defined by transition $\trans_i$. In the modified
encoding we separate elements representing states and transitions. In
a \CTwoLinOne formula we say that every element representing a state
(with the exception of first and last one) has precisely one $s$ edge
to- and from- an element representing transition. Similarly, every
transition element has precisely one $s$ edge to- and from- a state
element. Since we still require that $s(x,y)$ implies $x \Lin y$, it
follows that a state element $x$ (\resp a transition element) is
connected to a transition element $y$ (\resp a state element) by an
$s$ edge if and only if $y$ is the successor of $x$ \wrt $\Lin$ (we
have seen a formula that specifies a similar property in
Section~\ref{sec:twoorders}). The formula does not constrain in any
way $s$ edges between two transition elements. Thus we may expand the
formula with conjuncts that use $s$ edges to match an increase of a
counter in a transition with a decrease of the same counter in some
following transition. This leads to the following corollary.

\begin{theorem}
The finite satisfiability problem for \CTwoLinOne is at least as hard as
the emptiness problem for multicounter automata.
\end{theorem}

Note that by Remark~\ref{rem:succ} we have the
following corollary.
\begin{corollary}
  The finite satisfiability problem for \CTwoTwoPrec, where $\prec$ is
  interpreted as an acyclic relation, is at least as hard as the
  emptiness problem for multicounter automata.
\end{corollary}

\section{Satisfiability of \texorpdfstring{\CTwoOneLinS}{C2[...]}} 
\label{sec:upperc2linS}
In this section we show that the finite satisfiability problem of
\CTwoOneLinS is decidable. Here again the algorithm consists of three
parts: we bring the input formula to a normal form and then we check
its local and global consistency. Local consistency is checked using
\emph{frames} that store information about possible types in a
structure. Global consistency is checked using multicounter
automata. Compared to the previous section, both types and automata
are much more complicated.

\subsection{Normal form of \texorpdfstring{\CTwo}{C2} formulas}
For a natural number $n$ denote by $\underline{n}$ the set
$\{1,\ldots,n\}$.  We will assume that input \CTwoOneLinS formula
$\varphi$ is in a normal form
\begin{eqnarray}\label{nf}
  \varphi = & \forall{x}\forall{y}.(\alpha(x,y) \vee x=y)\wedge 
  & \bigwedge_{h\in \underline{m}}
  \forall{x}\exists^{=1}{y}.(f_h(x,y)\wedge x\neq y)
\end{eqnarray}
where $\alpha$ is a~quantifier-free formula with unary and binary
predicate symbols and $f_1,\ldots, f_{m}$ are binary predicates. By a
routine adaptation of transformation in~\cite{GradelO99} we may
convert each \CTwo formula to an exponentially larger \CTwo formula
$\varphi'$ in normal form, such that $\varphi$ and~$\varphi'$ are
equisatisfiable (on structures of cardinality $>1$).  
\begin{remark}
  In Section~\ref{sec:linsucc} we use a notion of normal forms for
  \CTwo formulas with only polynomial blowup. Here, for simplicity of
  presentation, we decided to employ the one with exponential
  blowup. Since there is no elementary upper bound on the complexity
  of the problem to which we reduce our logic, the construction in the
  present section would not benefit from the usage of a more succinct
  normal form.
\end{remark}

\subsection{2-types and message types}
To be able to reason about local consistency of a formula we introduce
several (standard in finite-model theory) notions of types, often
following notations from~\cite{IPH10}.  Intuitively, a 2-type defined
below is a generalisation of an order formula from
Section~\ref{sec:linsucc}. It contains enough information to check
whether the conjunct $\alpha$ from a formula in normal form (\ref{nf})
is locally true and if a given element is a witness for another
element \wrt a conjunct $\forall{x}\exists^{=1}{y}.(f_h(x,y)\wedge
x\neq y)$. Star types defined in Section~\ref{sec:startypes}
generalise full types from Section~\ref{sec:linsucc} and contain
enough information to check if an element has the right number of
witnesses.

Fix a finite signature $\Sigma$.  A \emph{2-type} is a~maximal
consistent conjunction of atomic and negated atomic formulas over
$\Sigma$ involving only the variables $x$ and $y$ and satisfying three
additional restrictions: first, it contains $\neg (x=y)$; second,
whenever it contains $\Succ(x,y)$ or $\Succ(y,x)$, it also contains
respectively $x\Lin y$ or $y\Lin x$; and third, it contains either
$x\Lin y$ or $y\Lin x$, but not both.  We will identify a 1-type $\pi$
(\resp\ a 2-type $\tau$) with the set of positive atomic formulas occurring
in $\pi$ (\resp\ in $\tau$). Each 2-type $\tau(x,y)$ uniquely determines two
1-types of $x$ and $y$, respectively, that we denote $\tpf{\tau}$ and
$\tps{\tau}$. For a 2-type $\tau$, the 2-type obtained by swapping the
variables $x$ and $y$ is denoted $\tau^{-1}$.  The symbol
$\mathcal{T}(\Sigma)$ denotes the set of of 2-types over $\Sigma$.

For a~structure $\str A$ over the signature $\Sigma$ and an element
$e\in \str A$, $\tpu (e)$ denotes the~unique 1-type $\pi\in \Pi(\Sigma)$
such that $\str A\models\pi(e)$. Similarly, for $e_1,e_2\in \str A$,
$\tpu(e_1,e_2)$ is the unique 2-type $\tau\in \mathcal{T}(\Sigma)$
such that $\str A\models\tau(e_1,e_2)$. If $\str
A\models\tau(e_1,e_2)$, we say that $e_1$ \emph{emits} the type $\tau$
and $e_2$ \emph{accepts} it and that $\tau$ \emph{originates} in
$e_1$. A 1-type $\pi$ (\resp 2-type $\tau$) is \emph{realised} in
$\str A$ if $\pi = \tp{\str A}(e)$ (\resp $\tau = \tp{\str
  A}(e_1,e_2)$) for some $e\in \str A$ (\resp $e_1,e_2\in \str A$,
with $e_1\neq e_2$).  The symbols $\Pi(\str A)$ and $\mathcal{T}(\str A)$
denote respectively the set of 1-types and the set of 2-types over
$\Sigma$ realised in $\str A$.
A 1-type $\kappa \in \Pi(\Sigma)$ that has only one realisation in a
structure $\str A$ is said to be a \emph{royal 1-type} in $\str A$.  If
an element $e$ of $\str A$ realises a {royal 1-type} then it is said to
be a \emph{king} in $\str A$. Any structure may have multiple kings.

If $\Sigma$ is a relational signature and $\overline{f} =
f_1,\ldots,f_m$ is a sequence of distinct binary predicates in
$\Sigma$, then the pair $\tuple{\Sigma,\overline{f}}$ is called a
\emph{classified signature}.  Let $\tuple{\Sigma,\bar{f}}$ be a
classified signature and let $\tau(x,y)$ be a 2-type over $\Sigma$. We
say that $\tau$ is a~\emph{message type} over $\tuple{\Sigma,\bar{f}}$
if $f(x,y)\in \tau(x,y)$ for some predicate $f$ in
$\bar{f}$. Predicates in $\bar{f}$ will be called \emph{message}
predicates. Given a structure $\str A$ over the signature
$\tuple{\Sigma,\bar{f}}$ and an element $a\in \str A$, we want to
capture message types connecting $a$ to other elements of $\str A$ and
all 2-types connecting $a$ to kings of $\str A$.  We first define the
set of all these 2-types. If $\Kings$ is a set of royal 1-types from
$\str A$, then denote by $\tau(\Kings,\Sigma,\bar{f})$ the set of all
2-types $\mu$, such that $\mu$ is a message type over
$\tuple{\Sigma,\bar{f}}$ or $\otp{2}(\mu)\in \Kings$. A 2-type from
$\tau(\Kings,\Sigma,\bar{f})$ is called an \emph{essential type}. If
$\tau$ is an essential type (\resp a message type) such that
$\tau^{-1}$ is also an essential type (\resp a message type) then we
say that $\tau$ is an \emph{invertible essential type} (\resp
\emph{invertible message type}).  On the other hand, if $\tau$ is a
2-type such that neither $\tau$ nor $\tau^{-1}$ is an essential type,
then we say that $\tau$ is a \emph{silent type}.

Given a structure $\str A$ over a~classified signature
$\tuple{\Sigma,\overline{f}}$ and a message type $\tau$, if $\str A\models
\tau(e_1,e_2)$ then $e_2$ is called a \emph{witness for $e_1$} in
$\str A$. It follows that if $\tau$ is an invertible message type then
also $e_1$ is a witness for $e_2$. This is because $\str A\models
\tau^{-1}(e_2,e_1)$ and $\tau^{-1}$ is an (invertible) message type.
From now on we assume that the classified signature of $\varphi$ is
$\tuple{\Sigma,\bar{f}}$, that is, $\Sigma$ contains unary and binary
predicates from $\varphi$ and $\bar{f} = f_1,\ldots f_m$ is the
sequence of binary predicates occurring under existential quantifiers
in $\varphi$.

Since we consider predicates of arity at most 2, a structure $\str A$
can be seen as a complete directed graph, where nodes are labelled by
1-types and edges are labelled by 2-types. Thus to define such a
structure it is enough to define 1-types of its elements and 2-types
of all pairs of elements provided that the projections of 2-types onto
1-types coincide with these 1-types and that for each pair
$\tuple{e_1,e_2}$ of elements connected by a 2-type $\mu$ the pair
$\tuple{e_2,e_1}$ is connected by the 2-type $\mu^{-1}$. In the rest
of the paper we use the notions of nodes and elements interchangeably.

\subsection{Normal structures}
Now, to simplify the reasoning, we restrict the class of models that
we consider. Intuitively, besides some simple sanity conditions, we
want to avoid non-royal 1-types with small number of occurrences. In
other words, each 1-type either should be royal or it should have many
occurrences in a structure. Here ``many'' means that we can always
find a representative of this 1-type that is not a witness in any
conjunct of the form $ \forall{x}\exists^{=1}{y}\ldots$ and thus can
be connected to any other non-royal element with a silent 2-type.
\begin{definition}\label{def:normal-structure}
  A finite structure $\str A\in \ClassO$ over a classified signature
  $\tuple{\Sigma,\bar{f}}$ is \emph{normal} if
\begin{enumerate}
\item both the smallest and the largest elements \wrt $\Lin^{\str A}$ are
  kings in $\str A$,\label{def:normal-first}
\item for every two non-royal elements $e_1, e_2 \in \str A$ satisfying
  $e_1 \Lin^{\str A} e_2$ there exist two elements $e'_1, e'_2 \in
  \str A$ such that $e'_1 \Lin^{\str A} e'_2$, $\tp{\str A}(e_1) =
  \tp{\str A}(e'_1)$, $\tp{\str A}(e_2) = \tp{\str A}(e'_2)$, and
  $\tp{\str A}(e'_1,e'_2)$ is a silent 2-type, \label{def:normal-second}
\item for every node $e\in \str A$ and $f\in\bar{f}$ we have $|\{e'\in
  \str A\mid \str A \models f(e,e')\}| = 1$,
  and \label{def:normal-third}
\item for every $e_1,e_2\in \str A$ if $\Succ^{\str A}(e_1,e_2)$ then
  $\tp{\str A}(e_1,e_2)$ is an invertible essential type.\label{def:normal-fourth}
\end{enumerate}
\end{definition}
Lemma~\ref{lem:normalstr} below says that when dealing with models of
\CTwoOneLinS formulas, we may restrict to normal structures. To prove
it we will need some technical lemmas.  If $\str A$ is a model of a
formula in \CTwoOneLinS and $S$ is a~set containing some elements of
$\str A$, then by $\min_{\str A}(S)$ we denote the smallest element of
$S$ \wrt the linear order $\Lin^{\str A}$. Similarly, by $\max_{\str
  A}(S)$ we denote the largest element of $S$ \wrt $\Lin^{\str A}$.

Let $\str A$ be a structure over $\tuple{\Sigma,\bar{f}}$ and $\pi$ be
a 1-type realised in $\str A$.  Let $\{e_i\}_{i=1}^k$ be the sequence
of all elements of $\str A$ that realise $\pi$, and such that $e_i
\Lin^{\str A} e_{i+1}$ for $i\in \{1,\ldots,k-1\}$. Define $S(\str
A,\pi) = \{e_i \mid i\in \{1,\ldots, \min(k,2m+1)\}\}$ and $L(\str
A,\pi) = \{e_i \mid i\in \{\max(1,k + 1 - (2m+1)),\ldots,
k\}$. Intuitively, $S(\str A,\pi)$ consists of $\min(k,2m+1)$ smallest
(\wrt $\Lin^{\str A}$) nodes that realise 1-type $\pi$, and $L(\str
A,\pi)$ consists of $\min(k,2m+1)$ largest nodes that
realise~$\pi$. Recall that here $m$ is the number of predicates in
$\bar{f}$. If $\pi,\pi'$ are two 1-types realised in $\str A$ then we
say that a pair $\tuple{\pi,\pi'}$ is \emph{correct} \wrt structure
$\str A$ if $|S(\str A,\pi)| = 2m+1$, $|L(\str A,\pi')| = 2m+1$ and $
\max_{\str A}(S(\str A,\pi)) <^{\str A} \min_{\str A}(L(\str
A,\pi'))$. A pair $\tuple{\pi,\pi'}$ is \emph{incorrect} \wrt $\str A$
if it is not correct \wrt $\str A$.

\begin{lemma}\label{lem:silent-types-exist}
  Let $\str A$ be a structure over $\tuple{\Sigma,\bar{f}}$. Let $\pi$
  and $\pi'$ be 1-types such that the pair $\tuple{\pi,\pi'}$ is correct
  \wrt $\str A$. Then, for some distinct elements $e, e'\in \str A$
  the 2-type $\tau$ connecting $e$ to $e'$ is silent, $\tpf{\tau} =
  \pi$, $\tps{\tau} = \pi'$ and $(x \Lin y) \in \tau$.
\end{lemma} 
\begin{proof}
  Since the pair $\tuple{\pi,\pi'}$ is correct, there exist two sequences
  of $\str A$'s elements $\{e_i\}^{2m+1}_{i=1}$ and $\{e'_i\}_{i
    =1}^{2m+1}$ such that $\max_{\str A}(\{e_i\}_{i=1}^{2m+1})
  \Lin^{\str A} \min_{\str A}(\{e'_i\}_{i=1}^{2m+1})$, $\tp{A}(e_i) =
  \pi $ and $\tp{A}(e'_i) = \pi'$ for $i\in{1,\ldots,(2m+1)}$.

  Consider the substructure of $\str A$ generated by (directed) edges
  connecting elements of $\{e_i\}_{i=1}^{2m+1}$ to elements of
  $\{e'_i\}_{i=1}^{2m+1}$. The number of these edges is $(2m+1)^2$.
  Since each element needs at most $m$ witnesses (and none of the
  types $\pi, \pi'$ is royal), there are at most $2(2m+1)m$ edges
  among them that are labelled with essential types.  Since $2(2m+1)m <
  (2m+1)^2$, at least one of these edges is labelled with a~silent
  type. Take such an edge. It connects some element $e\in
  \{e_i\}_{i=1}^{2m+1}$ to some $e'\in \{e'_i\}_{i=1}^{2m+1}$. Let
  $\tau = \tp{\str A}(e,e')$. Clearly, $\tau$ is a silent 2-type,
  $\tpf{\tau} = \pi$, $\tps{\tau} = \pi'$ and $(x \Lin y) \in \tau$.
\end{proof}

In every finite structure, which we are dealing with, the largest
(\wrt $\Lin$) node is always present. Thus predicate $\Succ$ cannot be
among message predicates of any normal structure as it would
violate Condition~\ref{def:normal-third} of
Definition~\ref{def:normal-structure}. For similar reasons $\Lin$
cannot be one of these predicates. Therefore, to satisfy
Conditions~\ref{def:normal-third} and~\ref{def:normal-fourth}, we need
the following lemma.
\begin{lemma}\label{lem:normalstr4}
  Let $\str A$ be a structure over a classified signature
  $\tuple{\Sigma,\bar{f}}$. By interpreting two additional binary
  message-predicates we can expand $\str A$ to a~structure $\str B$
  such that for every $e_1,e_2\in \str B$ if $\Succ^{\str B}(e_1,e_2)$
  then $\tp{\str B}(e_1,e_2)$ is an invertible essential type.
\end{lemma}
\begin{proof}
  Let $e_1,\ldots,e_k$ be all nodes of $\str A$ sequenced in order
  $<^{\str A}$.  Let $f_{\mathit{back}}$ and $f_{\mathit{forth}}$ be
  two fresh binary predicates. We add them to $\bar{f}$ and expand the
  structure $\str A$ by putting $f^{\str B}_{\mathit{forth}} =
  \{\tuple{e_i,e_{i+1}}\mid i<k\} \cup \{\tuple{e_k,e_1}\}$ and
  $f^{\str B}_{\mathit{back}} = \{\tuple{e_{i+1},e_{i}}\mid i<k\} \cup
  \{\tuple{e_1,e_k}\}$. Then for all $i\in \underline{k-1}$ we have
  $\tp{\str B}(e_i,e_{i+1})$ is an invertible essential type.
\end{proof}

\begin{lemma}\label{lem:normalstr}
  Let $\varphi$ be a \CTwoOneLinS formula in normal form, over a
  classified signature $\tuple{\Sigma,\bar{f}}$. If $\varphi$ is finitely
  satisfiable then there exists a signature
  $\tuple{\Sigma',\bar{f'}}$ such that $\Sigma \subseteq \Sigma'$ and
  $\bar{f}\subseteq\bar{f'}$ and a finite normal
  $\tuple{\Sigma',\bar{f'}}$-structure $\str B$ such that $\str B
  \models \varphi$. Moreover, $|\Sigma'|$ is polynomial in $|\Sigma|$
  and $|\bar{f'}|=|\bar{f}|+2$.
\end{lemma}

\begin{proof}
  Let $\str A$ be a model of $\varphi$. By
  Lemma~\ref{lem:normalstr4} we may assume that
  Condition~\ref{def:normal-fourth} of
  Definition~\ref{def:normal-structure} is satisfied.
  Condition~\ref{def:normal-third} of this definition is also true, as
  $\varphi$ is in normal form.  The following algorithm constructs
  $\str B$ from $\str A$ in a sequence of steps numbered by a natural
  number $j$. In each step the algorithm interprets fresh unary
  predicates and does not change the interpretation of binary
  predicates thus not violating Conditions~\ref{def:normal-third}
  and~\ref{def:normal-fourth}.

  Initially $j=0$, structure $\str B_0$ is $\str A$ where the smallest
  and the largest node is turned into king by interpreting two fresh
  unary predicates, and $\mathit{Pairs} = \{\tuple{\pi,\pi'} |
  \pi$ and $\pi'$ are non-royal 1-types realised in $\str
    B_0\}$.
\begin{enumerate}
\item while there exists a pair $\tuple{\pi,\pi'} \in \mathit{Pairs}$
  which is incorrect \wrt $\str B_j$, execute steps
  (\ref{coloring:second})--(\ref{coloring:third})\label{coloring:first}
\item construct structure $B_{j+1}$ by labelling all nodes of $B_j$
  from $S(\str B_j,\pi) \cup L(\str B_j,\pi')$ by fresh unary
  predicates to turn them into kings of
  $B_{j+1}$, \label{coloring:second}
\item remove pair $\tuple{\pi,\pi'}$ from $\mathit{Pairs}$, increment
  $j$, \label{coloring:third}
\item put $\str B = \str B_j$ and terminate the algorithm
\end{enumerate}
Clearly, the above algorithm terminates after at most
$|\mathit{Pairs}|$ steps, \ie after $O(2^{2|\Sigma|})$ steps.  At each
step the algorithm creates at most $2(2m+1)$ new kings, so the total
number of kings created by the algorithm is $2(2m+1)|\mathit{Pairs}|$,
which is $O(2(2m+1)2^{2|\Sigma|})$. The number of fresh unary
predicates required to distinguish that number of kings is logarithmic
\wrt that latter number. Therefore $|\Sigma'|$ is polynomial in
$|\Sigma|$. Since $\str A$ models $\varphi$, also $\str B$ models~$\varphi$.

We now show that $\str B$ is normal. Condition~\ref{def:normal-first}
in Definition~\ref{def:normal-structure} is satisfied, because it is
satisfied by $\str B_0$ and all kings of $\str B_0$ are left untouched
during the construction of $\str B$ from $\str B_0$. 

In order to show that Condition~\ref{def:normal-second} in
Definition~\ref{def:normal-structure} is also satisfied, take non-royal
elements $e_1, e_2 \in \str B$ such that $e_1 \Lin^{\str B} e_2$. Let
$\pi = \tp{\str B}(e_1)$ and $\pi' = \tp{\str B}(e_2)$. Since $\pi$
and~$\pi'$ are non-royal 1-types realised in $\str B$, they are also
non-royal 1-types realised in $\str A$. Thus, at the beginning of the
algorithm we have $\tuple{\pi,\pi'} \in \mathit{Pairs}$. Moreover,
tuple $\tuple{\pi,\pi'}$ was never selected at
step~(\ref{coloring:first}) of the algorithm. Therefore pair
$\tuple{\pi,\pi'}$ is  correct \wrt $\str B$. By
Lemma~\ref{lem:silent-types-exist} we obtain the desired conclusion.
Thus $\str B$ is a normal structure.
\end{proof}

\subsection{Star types}
\label{sec:startypes}
Given a structure $\str A$ over a signature $\tuple{\Sigma,\bar{f}}$
and an element $a\in \str A$, we want to capture essential 2-types
emitted from $a$ to other elements of $\str A$. For this reason we
introduce \emph{star types}. Intuitively a~star type of an element $a$
in a~structure describes a~small neighbourhood of $a$ and contains
enough information to check whether $a$ has the right number of
witnesses.

\begin{definition}[Star type in $\str A$]\label{def:startype}
  Let $\str A$ be a normal structure over $\tuple{\Sigma,\bar{f}}$,
  and let $a$ be an element of $\str A$.  Let $\Kings = \{\kappa \mid
  \kappa \text{ is a royal type in } \str A \}$.  The \emph{star type} of
  $a$ in $\str A$, denoted $\stp{\str A}(a)$ is a pair
  $\sigma=\tuple{\pi, \mathcal{T}}$ where $\pi=\tpu(a)$ and
  $\mathcal{T}$ is the set of essential types originating in $a$:
  \[\mathcal{T} = \{\mu\in \tau(\Kings,\Sigma, \bar{f})\mid
  \text{$\tpu(a,b)=\mu$ for some $b\in \str A$}\}.\]
  We denote 
  the type $\pi$ by $\pi(\sigma)$.  We say that a~2-type $\mu$
  \emph{occurs} in $\sigma$, written $\mu \in \sigma$, if $\mu \in
  \mathcal{T}$. We write $\sigma - \mu$ for the star-type $\sigma' =
  \tuple{\pi,\mathcal{T}\setminus \{\mu\}}$. When $S$ is a set of
  2-types we write $\sigma \setminus S$ to denote the star type
  $\tuple{\pi,\mathcal{T}\setminus S}$.
\end{definition}

Observe that in the definition above $\sigma$ satisfies the conditions
\begin{enumerate}
\item \label{eq:tp0} for all $\mu_1,\mu_2 \in \sigma$ if $f(x,y)\in \mu_1$ and
  $f(x,y)\in \mu_2$ for some $f\in \bar{f}$ then $\mu_1 =
  \mu_2$, 
\item \label{eq:tp1} $\mu \in \sigma$ implies $\tpf{\mu}=\pi$ for all
  $\mu \in \tau(\Kings,\Sigma, \bar{f})$,
\item \label{eq:tp2} $|\{\mu\in\sigma\mid \otp{2}(\mu)=\kappa\}| = 1 $
  for all $\kappa \in \Kings$ such that $\kappa\neq \pi$,
\item \label{eq:tp3} $|\{\mu\in\sigma\mid \otp{2}(\mu)=\pi\}| = 0 $ if
  $\pi \in \Kings$.
\end{enumerate}
The first of these conditions is obvious, as normal structures
emit precisely one edge $\tau$ with $f(x,y)\in \tau$ for $f\in
\bar{f}$. The second one says that all 2-types originating in $a$ have
the same 1-type of the origin, namely the 1-type of $a$. The third one
says that for all kings $k$ (in $\str A$) the element $a$ is connected
with $k$ by exactly one 2-type (provided that $k\neq a$). The last
condition says that if $a$ is a~king then it is not connected with
itself by any 2-type (recall that 2-types connect different elements).

\begin{definition}\label{def:startype2}        
  A \emph{star type} over the set of 2-types
  $\tau(\Kings,\Sigma,\bar{f})$ is any pair of the form $\tuple{\pi,
  \mathcal{T}}$ satisfying Conditions~\ref{eq:tp0}--\ref{eq:tp3}
above.  A structure $\str A$ is said to realise a star type $\sigma$
if $\stp{\str A}(a) = \sigma$ for some $a\in \str A$.
\end{definition}

For a~given set of star types $\ST$, by $\pi(\ST)$ we denote the set
of 1-types $\{\pi(\sigma) \mid \sigma\in \ST\}$.  The set of 2-types
occurring in star types from $\ST$, that is the set $\{\mu \mid
\exists{\sigma\in \ST} \;\;\mu \in \sigma\}$ is denoted by
$\tau(\ST)$, and the set $\{\tuple{\pi,
  \mathcal{T'}}\mid\text{$\tuple{\pi, \mathcal{T}}\in\ST$ for some
  $\mathcal{T}$ satisfying $\mathcal{T'}\subseteq \mathcal{T}$}\}$ by
$\Partial{\ST}$.  Elements of $\Partial{\ST}$ are called
\emph{partial star types}. A partial star type is said to be
\emph{empty} if it is of the form $\tuple{\pi, \emptyset}$.

\subsection{Frames}
We now introduce finite and small structures called frames. Frames
contain, among others, the information about all 2-types and all star
types that may occur in a~structure. This is enough to check whether
the universal subformula $\alpha$ of a~formula in normal
form~(\ref{nf}) is true and if all elements have the right number of
witnesses. Intuitively, a~correctly guessed frame confirms that
a~formula is locally consistent (see Definition~\ref{def:models}).

\begin{definition}[Frame]\label{def:frame}
  Let $\tuple{\Sigma, \bar{f}}$ be a classified signature, $\Kings$ be
  a set of 1-types over $\Sigma$, let $\ST$ be a set of star types
  over $\tau(\Kings,\Sigma,\bar{f})$ and let $\Xi$ be a set of silent
  2-types over $\tuple{\Sigma,\bar{f}}$.  The tuple
  $\tuple{\Kings,\ST,\Xi,\Sigma,\bar{f}}$ is called \emph{a frame} if
  the following conditions are satisfied
  \begin{enumerate}
  \item \label{frame:invertible} for each 2-type $\tau \in
    \tau(\ST)\cup \Xi$ if $\Succ(x,y)\in \tau$ or $\Succ(y,x)\in \tau$
    then $\tau$ is an invertible essential type,
  \item \label{frame:first} there exists exactly one star type
    $\sigma_\mathit{first}\in \ST$ such that for all $\tau\in
    \sigma_\mathit{first}$ we have $\Succ(y,x)\not\in\tau$,
  \item \label{frame:last} there exists exactly one star type
    $\sigma_\mathit{last}\in \ST$ such that for all $\tau\in
    \sigma_\mathit{last}$ we have $\Succ(x,y)\not\in\tau$,
  \item \label{frame:4} for each $\kappa \in \Kings$ there exists
    exactly one $\sigma \in \ST$ such that
    $\pi(\sigma)=\kappa$, and
  \item \label{frame:44} for each star type $\sigma \in \ST$
    and each 2-type $\mu$, if $\mu \in \sigma$ then
    $\otp{2}(\mu) \in \pi(\ST)$. 
  \end{enumerate}
\end{definition}

Frames are intended to describe local configurations in normal
structures $\str A$.  The set $\Kings$ contains all royal 1-types of
$\str A$, $\ST$ contains all star types of $\str A$ and the set $\Xi$
contains all silent 2-types realised in $\str
A$. Condition~\ref{frame:invertible} says that every node in $\str A$
is connected to its successor and predecessor by invertible essential
types. Conditions~\ref{frame:first} and~\ref{frame:last} say that
there are unique star types for the first and the last node in $\str
A$.  Note that Conditions~\ref{frame:invertible}--\ref{frame:last} are
true in every normal structure.  Condition~\ref{frame:4} says that
each king has exactly one star type.  Condition~\ref{frame:44} ensures
that if a neighbour of a node in a structure has some 1-type $\pi$,
then there exists a star type $\sigma\in \ST$ such that $\pi\in
\pi(\ST)$. The last two conditions hold in every relational
structure.

Intuitively, we want to check finite satisfiability of a~\CTwo formula
$\varphi$ by guessing a~right frame. ``Right'' means here that two
conditions must be satisfied. First, the frame should be locally
consistent with $\varphi$. This means that every 2-type occurring in
the frame entails the subformulas of $\varphi$ of the form $\forall
x\forall y\ldots$, and that the number of witnesses in every star type
is correct. This is formalised in the following definition.  Second,
the frame should be globally consistent in the sense that there exists
a structure that fits this frame -- this is formalised in
Definition ~\ref{def:fits}.

\begin{definition}[$\Frame \models \varphi$]\label{def:models}
  Consider a frame  $\Frame = \tuple{\Kings,\ST,\Xi,\Sigma,\bar{f}}$ and
   a \CTwoOneLinS formula $\varphi$  in normal form~(\ref{nf})
  over  $\tuple{\Sigma, \bar{f}}$. We say that $\Frame$
  \emph{satisfies} $\varphi$, in symbols $\Frame\models \varphi$, if
  \begin{itemize}
  \item for each 2-type $\mu\in \Xi\cup \tau(\ST)$, the formula
    $\alpha$ is a consequence of $\mu$ and of $\mu^{-1}$, that is
    $\models \mu \rightarrow \alpha$ and $\models \mu^{-1} \rightarrow
    \alpha$, where $\mu$ is seen as conjunction of literals, and
  \item for each $\sigma\in \ST$ and $h\in \underline{m}$ we have
$|\{\mu\in\sigma\mid f_h(x,y)\in \mu\}| = 1$.
  \end{itemize}
\end{definition}

\begin{definition}\label{def:fits}
  Let $\tuple{\Kings,\ST,\Xi,\Sigma,\bar{f},c}$ be a frame, and $\str
  A$ a structure over $\tuple{\Sigma, \bar{f}}$. We say that $\str A$
  \emph{fits the frame} $\Frame$ if
  \begin{itemize}
  \item the set of royal 1-types realised in structure $\str A$ is
    $\Kings$, and
  \item the set of all silent types realised in $\str A$ is a
    subset of $\Xi$, and
  \item the set of star types of $\str A$ is a subset of
    $\ST$, in symbols $\stp{\str A}(\str A)
    \subseteq \ST$.
  \end{itemize}
\end{definition}

The following proposition reduces the finite satisfiability problem of
$\CTwoOneLinS$ to the problem of existence of a~structure in $\ClassO$
that fits a~given frame. This is the reduction that splits the problem
into local and global consistency.

\begin{proposition}\label{prop:red}
  Let $\varphi$ be a \CTwoOneLinS formula in normal form over a
  classified signature $\tuple{\Sigma, \overline{f}}$, where
  $\{\Lin,\Succ\} \subseteq \Sigma$. Let $\str A$ be a structure in
  $\ClassO$.
  \begin{enumerate}
  \item If $\str A$ is normal and $\str A\models \varphi$ then there
    exists a frame $\Frame$, such that $\str A$ fits $\Frame$ and
    $\Frame\models \varphi$.
  \item If there exists a frame $\Frame$ such that $\str A$ fits
    $\Frame$ and $\Frame\models \varphi$ then $\str A\models
    \varphi$.
  \end{enumerate}
\end{proposition}

\begin{proof}
  For the proof of the first statement, assume that $\str A$ is normal
  and $\str A\models \varphi$. Let $\Kings$ be the set of royal 1-types
  realised in $\str A$, let $\ST$ be the set of star types of $\str A$
  and let $\Xi$ be the set of all silent types realised in
  $\str A$. The facts that the tuple
  $\Frame=\tuple{\Kings,\ST,\Xi,\Sigma,\bar{f},c}$ forms a frame,
  $\Frame\models \varphi$ and $\str A$ fits $\Frame$ are immediate,
  once Definitions~\ref{def:frame}, \ref{def:models} and
  \ref{def:fits} are unravelled.

  For the proof of the second statement, let $\Frame$ be a frame such
  that $\Frame\models \varphi$ and let $\str A$ be a structure such
  that $\str A$ fits $\Frame$. Since $\varphi$ is in normal form,
  it is of the form~(\ref{nf}). Let $\mu$ be any 2-type realised in
  $\str A$. Then either $\mu$ is a silent type or it occurs
  in some star type realised in $\str A$. In any case, by Definition
  \ref{def:fits} we have that $\mu\in \Xi\cup \tau(\ST)$.  By
  Definition \ref{def:models} it follows that $\models \mu \rightarrow
  \alpha$.  So $\str A\models \forall{x}\forall{y}.(\alpha \vee x=y)$.
  Since $\str A$ fits  $ \mathcal{F}$ and $\mathcal{F}\models
  \varphi$, it also follows that for each star type $\sigma$ realised
  in $\str A$ and each $h$ such that $1\leq h \leq m$ we have
  $|\{\mu\in\sigma\mid f_h(x,y)\in \mu\}| = 1$, and thus $\str
  A\models \bigwedge_{h\in \underline{m}}
  \forall{x}\exists^{=1}{y}.(f_h(x,y)\wedge x\neq y)$.  Hence $\str
  A\models \varphi$ as required.
\end{proof}

\subsection{High-level multicounter automata}
In the rest of this section we show how to decide, for a~given frame
$\Frame$, whether there exists a~structure in $\ClassO$ that fits this
frame.  This will be done by a~reduction to the emptiness problem for
multicounter automata.
 
We now introduce a syntactic extension to multicounter automata that
we call \emph{High-level MultiCounter Automata} (\textbf{HMCA}).  The
idea is to specify transitions of an automaton as programs in a
higher-level imperative language with conditionals, loops and arrays,
which leads to clearer exposition of reachability problems. A
transition in a High-Level MCA is a sequence $\Delta$ of actions;
each action in turn may update and test finite-domain variables, and
conduct conditional or loop instructions depending on results of these
tests. A transition may also increment or decrement, but not test the
value of, counters, which are the only infinite-domain variables of
the automaton.

In the following subsections we give formal syntax and semantics of
high-level multicounter automata and we prove that HMCA can be
compiled to multicounter automata. If the reader is familiar with
finite-state machines and how they handle finite information, it
should be more or less clear how to translate HMCA to MCA. In such
a~case we suggest to skip Sections~\ref{sec:hmca_semantics} and
\ref{sec:hmca_compilation} and move directly to
Section~\ref{sec:global}.

\subsubsection{Syntax of HMCA}
Formally, an HMCA is a tuple $H = \tuple{\Var, \Vect, \Type, \Delta,
  \rho_I, P_F, E}$. The set $\Var$ consists of variables $v$ to be
interpreted in the finite domain $\Type(v)$. We may think of $\Var$ as
a~declaration of finite-domain variables of the program.  The set
$\Vect$ corresponds to a~declaration of arrays, it consists of
variables $\overrightarrow{vec}$ to be interpreted as vectors of
natural numbers indexed by elements of some finite set $\mathbb{A}$,
where $\Type(\overrightarrow{vec}) = \mathbb{A}\rightarrow\Nat$. We
will refer to the index set $\mathbb{A}$ as the domain
$\mathrm{Dom}(\overrightarrow{vec})$. The $\Type$ function assigns to
every variable in $\Var\cup \Vect$ its type.  The sequence of actions
$\Delta$ is the actual program built from actions defined below.  The
starting state of $H$ is $\rho_I$, the set of accepting states is
$P_F$. The set $E$ is a subset of $\{\tuple{\vec{v},a} \mid\text{
  $\vec{v}\in\Vect,\; a\in\mathrm{Dom}(\vec{v})$}\}$, and is used in
the acceptance condition explained later.

We now define a set of actions $\alpha$ that constitute transitions in
\textbf{HMCA}. The simplest action is an \emph{assignment} of the form
$v:=\Expr$, where $v$ is a variable of some domain $\mathbb{A} =
\Type(v)$, and $\Expr$ is an expression built from variables from
$\Var$, constants from appropriate domains and operators. An
\emph{operator} is any effectively computable function, \eg $\cup$,
$\cap$, $\setminus$ are operators of domain
$(2^{\mathbb{B}})^2\rightarrow 2^{\mathbb{B}}$ for any domain
$\mathbb{B}$; function $\pi: \ST(\K, \Sigma, \bar{f}) \rightarrow
\Pi(\Sigma)$ from Definition~\ref{def:startype} is also an operator,
provided that our finite domain contains $\ST(\K, \Sigma, \bar{f})$
and $\Pi(\Sigma)$. We silently extend the $\Type$ function to
constants by letting $\Type(a) = \mathbb{A}$ if $a\in \mathbb{A}$, and
to expressions, \eg $\Type(s_1 \cup s_2) = 2^{\mathbb{B}}$ if
$\Type(s_1) =2^{\mathbb{B}}$ and $\Type(s_2) =2^{\mathbb{B}}$. We
require that assignments $v:=\Expr$ are well typed, \ie that $\Type(v)
= \Type(\Expr)$. An \emph{atomic test} is of the form $\Expr_1 =
\Expr_2$ or $\Expr_3 \in \Expr_4$, where $\Expr_1$, $\Expr_2$,
$\Expr_3$, $\Expr_4$ are expressions such that $\Type(\Expr_1) =
\Type(\Expr_2)$ and $\Type(\Expr_4) = 2^{\Type(\Expr_3)}$. A
\emph{test} is an arbitrary Boolean combination of \emph{atomic
  tests}. Notice that tests do not use counters.  A
\emph{non-deterministic assignment} action is of the form \Guess $v\in
\Expr$ \with $\Test$, where $\Type(\Expr) = 2^{\Type(v)}$ and the
variable $v$ may occur in $\Test$. A conditional action is of the form
$\IIf\ \Test\ \IThen\ \alpha^{*}\ \IElse\ \alpha'^{*}\ \IEndIf$ or
$\IIf\ \Test\ \IThen\ \alpha^{*}\ \IEndIf$. A loop action is of the
form $\IWhile\ \Test\ \IDo\ \alpha^{*}\ \IEndWhile$.  An
\emph{incrementing action} (\resp \emph{decrementing action}) is of
the form $\inc(\vec{f}[\Expr])$ (\resp $\dec(\vec{f}[\Expr])$), where
$\Expr$ evaluates to an index of the array $\vec{f}$, that is,
$\vec{f} \in \Vect$, $\Type(\vec{f}) = \mathbb{A}\rightarrow \Nat$ and
$\Type(\Expr) = \mathbb{A}$. The remaining action, \Reject, simply
rejects current computation.

\subsubsection{Semantics of HMCA}
\label{sec:hmca_semantics}
Expressions and tests are evaluated in context of variable valuations.
A \emph{variable valuation} (also called a \emph{state}) is any
function $\rho$ that assigns to every finite-domain variable $v$ a
value $\llbracket v\rrbracket_{\rho} \in \Type({v})$. The semantics of
expressions is defined inductively: $\llbracket v \rrbracket_{\rho} =
\rho(v)$ for $v\in \mathbb{V}_\mathit{fin}$, $\llbracket \Expr_1
\bowtie \Expr_2 \rrbracket_{\rho} = \llbracket
\Expr_1\rrbracket_{\rho} \bowtie \llbracket\Expr_2 \rrbracket_{\rho}$
where $\bowtie\in \{\cup, \cap, \setminus\}$, and $\llbracket
f(\Expr_1,\ldots,\Expr_k)\rrbracket_{\rho} =
f(\llbracket\Expr_1\rrbracket_{\rho},\ldots,\llbracket\Expr_k\rrbracket_{\rho})$
where $f$ is an operator. In a similar way we define semantics of
tests.

A \emph{counter valuation} is any function $\vartheta$ that assigns
(a~sequence of) natural numbers to (arrays of) counters.  A
\emph{configuration} of \textbf{HMCA} $H$ is a pair
$\tuple{\rho,\vartheta}$ where $\rho$ is an variable valuation (\ie a
state) and $\vartheta$ is a counter valuation.  Actions transform
configurations. Most actions work only on variable valuations; the
exceptions are incrementing and decrementing of counters. With the
exception of the decrementing action, the semantics of actions is
self-explanatory; $dec(c)$ decrements the counter $c$ if it is
strictly positive and otherwise (if it is~$0$) it rejects the current
computation.

A \emph{run} of an HMCA $H$ is a sequence of configurations
$\tuple{\rho_1,\vartheta_1},\ldots \tuple{\rho_k,\vartheta_k}$ such
that $\tuple{\rho_{i+1},\vartheta_{i+1}}$ is obtained after executing
transition $\Delta$ in configuration $\tuple{\rho_{i},\vartheta_{i}}$,
for all $i\in \{1,\ldots, k-1\}$. A~single transition
$\sem{\tuple{\rho_i,\vartheta_i}}{\Delta}{\tuple{\rho_{i+1},\vartheta_{i+1}}}$,
is formalised on Figures~\ref{fig:sem-first} and~\ref{fig:sem-second}.
A run is \emph{accepting}, if it starts in an initial configuration
$\tuple{\rho_I,\vartheta_0}$ with $\rho_I$ being initial state and
$\vartheta_0$ assigning $0$ to all counters, and it ends in some
configuration $\tuple{\rho_F, \vartheta_F}$ with $\rho_F$ being
a~final state and $\vartheta_F$ assigning $0$ to all counters
specified in the set $E$ of final counters: $\vartheta_F(\vec{f})(a) =
0$ for every $\tuple{f,a} \in E$. The emptiness problem for high-level
multicounter automata is the question whether a given automaton $H$
has an accepting run.

For a counter valuation $\vartheta$, $\vec{v}\in\Vect$ and a function
$f$ with domain $\Type(\vec{v})$ by $\vartheta[\vec{v}\leftarrow f]$
we mean the valuation identical to $\vartheta$, with the exception
that $\vartheta(\vec{v}) = f$. In a similar way, for a variable
valuation $\rho$, $v\in\Var$ and $a\in\Type(v)$, we define the
variable valuation $\rho[v\leftarrow a]$.

Figures~\ref{fig:sem-first} and~\ref{fig:sem-second} present the
natural semantics of actions in a formal way, separately for actions
that operate on states and counters. The semantics of assignment,
conditional and loop actions is the expected one. \Eg assignment
$v:=\Expr$ updates current state $\rho$ by assigning to $v$ the value
$\llbracket\Expr\rrbracket_{\rho}$; conditional action $\IIf\ \Test\
\IThen\ \alpha^{*}\ \IEndIf$ executes $\alpha^{*}$ provided that
$\llbracket\Test\rrbracket_{\rho}$ is true, and skips to next action
otherwise.  A non-deterministic assignment \Guess $v\in \Expr$ \with
$\Test$ assigns to $v$ an arbitrary value $e\in
\llbracket\Expr\rrbracket_{\rho}$ such that $\Test$ holds when $v$
becomes $e$, \ie when $\llbracket\Test\rrbracket_{\rho[v\leftarrow
  e]}$ is true. Let $\Type(\vec{f}) = \mathbb{A} \rightarrow \Nat$ and
$e = \llbracket\Expr\rrbracket_{\rho}$. Incrementing action
$\inc(\vec{f}[\Expr])$ transforms a configuration
$\tuple{\rho,\vartheta}$ to
$\tuple{\rho,\vartheta[\vec{f}\leftarrow\vec{f'}]}$, where $f'(a) =
\vartheta(f)(a) + 1$ for $a=e$ and $f'(a) = \vartheta(f)(a)$ for
remaining elements $a\in \mathbb{A}$. Decrementing action
$\dec(\vec{f}[\Expr])$ works similarly, but it decrements the value
$\vartheta(f)(e)$, provided that it is positive (the computation
rejects if $\vartheta(f)(e) = 0$).  Finally, action $\Reject$ rejects
the computation unconditionally.
\begin{figure}[t]
\begin{gather*}
  \infer[]{\sem{\rho_1}{v:=\Expr}{\rho_2}}{\rho_2 =
    \rho_1[v\leftarrow \llbracket \Expr\rrbracket_{\rho_1}]}\\
  \infer[]{\sem{\rho_1}{\text{\Guess $v\in \Expr$ \with
        $\Test$}}{\rho_1[v\leftarrow e]}}{e\in \llbracket\Expr
    \rrbracket_{\rho_1} & \llbracket\Test\rrbracket_{\rho_1[v\leftarrow e]} = \true}\\
  \infer[]{\sem{\rho_1}{\IIf\ \Test\ \IThen\ \alpha^{*}\ \IElse\ \alpha'^{*}\ \IEndIf}{\rho_2}}{\llbracket\Test\rrbracket_{\rho_1} = \true & \sem{\rho_1}{\alpha^{*} }{\rho_2}}\\
  \infer[]{\sem{\rho_1}{\IIf\ \Test\ \IThen\ \alpha^{*}\ \IElse\ \alpha'^{*}\ \IEndIf}{\rho_2}}{\llbracket\Test\rrbracket_{\rho_1} = \false & \sem{\rho_1}{\alpha'^{*} }{\rho_2}}\\
  \infer[]{\sem{\rho_1}{\IWhile\ \Test\ \IDo\ \alpha^{*}\
      \IEndWhile}{\rho_2}}{ \llbracket\Test\rrbracket_{\rho_1} = \true
    & \sem{\rho_1}{\alpha^{*}}{\rho} & \sem{\rho}{\IWhile\ \Test\
      \IDo\ \alpha^{*}\ \IEndWhile}{\rho_2}
  }\\
  \infer[]{\sem{\rho_1}{\IWhile\ \Test\ \IDo\ \alpha^{*}\
      \IEndWhile}{\rho_1}}{
    \llbracket\Test\rrbracket_{\rho_1} = \false}
\end{gather*}
\caption{Semantics of actions that operate on state.  For these
  actions $\alpha$ we define
  $\sem{\tuple{\rho_1,\vartheta}}{\alpha}{\tuple{\rho_2,\vartheta}}$ if
  $\sem{\rho_1}{\alpha}{\rho_2}$.}\label{fig:sem-first}
\end{figure}

\begin{figure}
\begin{gather*}
  \infer[]{\sem{\tuple{\rho_1,\vartheta_1}}{\inc(\vec{f}(\Expr))}{\tuple{\rho_2,\vartheta_2}}}{e = \llbracket\Expr\rrbracket_{\rho_1} & \vartheta_2 =
    \vartheta_1[\vec{f}(e)\leftarrow \vec{f}(e)+1]}
\mbox{~~~}
\infer[]{\sem{\tuple{\rho_1,\vartheta_1}}{\alpha;\alpha^*}{\tuple{\rho_2,\vartheta_2}}}{\sem{\tuple{\rho_1,\vartheta_1}}{\alpha}{\tuple{\rho,\vartheta}} & \sem{\tuple{\rho,\vartheta}}{\alpha^{*}}{\tuple{\rho_2,\vartheta_2}}} 
\\
 \infer[]{\sem{\tuple{\rho_1,\vartheta_1}}{\dec(\vec{f}(\Expr))}{\tuple{\rho_2,\vartheta_2}}}{e = 
    \llbracket\Expr\rrbracket_{\rho_1} &
    \vartheta_1(\vec{f})(e) > 0 & \vartheta_2 =
    \vartheta_1[\vec{f}(e)\leftarrow \vec{f}(e)-1]}
\mbox{~~~}
\infer[]{\sem{\tuple{\rho_1,\vartheta_1}}{\epsilon}{\tuple{\rho_1,\vartheta_1}}}{ }
\end{gather*}
\caption{Semantics of incrementing and decrementing actions and semantics of actions composition.}
\label{fig:sem-second}
\end{figure}

\subsubsection{Compilation of HMCA}
\label{sec:hmca_compilation}
Intuitively, a~compilation of an HMCA to an~MCA consists in hiding the
control structures and finite-domain variables (including tests for
zero on finite-domain variables) in states of the constructed
MCA. Formally, we have the following proposition.

\begin{proposition}\label{prop:compile}
  Emptiness problem for \textbf{HMCA} is reducible to emptiness
  problem of multicounter automata, and is therefore decidable.
\end{proposition}

\begin{proof}
Let $H = \tuple{\Var, \Vect, \Type, \Delta, \rho_I, P_F, E}$ be an
HMCA. We define an MCA $M_H = \tuple{Q, C, R, \delta, q_I, F}$ such
that $H$ is non-empty if and only if $M_H$ is non-empty. Assume that
$\Delta$ consists of $k$ actions. Let $v_1,\ldots v_l$ be an arbitrary
enumeration of $\Var$ and $P = \Type(v_1)\times\ldots\times
\Type(v_l)$ be the set of all variable valuations.  Let $Q =
\{1,\ldots,k\}\times P $. Every state $q\in Q$ is a pair
$\tuple{i,\rho}$, which is intended to capture the number of an action
to be executed and values of variables $v_1,\ldots,v_k$ in the state
just before execution of the action. We use a shortcut $q(v_i)$ to
denote the $i$-th component of $k$-tuple $\rho$ and $q(1)$ to denote
its first component.  The set of counters $C$ of MCA $M_H$ is defined
as $C = \{c_{\vec{f}(e)} \mid\text{ $\vec{f} \in \Vect$ and $e\in
  \Type(\vec{f})$}\}$. The starting state $q_I$ of $M_H$ is
$\tuple{1,\rho_I}$. The set of final states $F = \{\tuple{1,\rho_F}
\mid \rho_F \in P_F\}$. The set $R = \{c_{\vec{f}(e)} \mid
\tuple{\vec{f},e}\in E\}$.  

Assume that each action $\alpha$ in
$\Delta$ is assigned a unique number $i\in \Nat$, interpreted as
line number in which $\alpha$ occurs. Below we will write $\alpha_i$
to denote action $\alpha$ occurring in line $i$. A \emph{control flow
  graph} (CFG) of $\Delta$ is a graph $(\{i\in\Nat\mid \alpha_i\in
\Delta\},\{\rightarrow, \xrightarrow{\true},
\xrightarrow{\false}\})$. Edges of CFG are defined as follows. For
every action $\alpha_i$, with the exception of loop and conditional
action, let $j$ be the line number of next action to be executed. If
no such action exists then we put $j=1$ denoting that $\Delta$ may
be reexecuted. We put $i \rightarrow j$. For loop or conditional
actions let $j_{\true}$ and $j_{\false}$ be numbers of actions to be
executed provided the action's test succeeds, \resp fails. As in
previous case, if $j_{\true}$ (\resp $j_{\false}$) are undefined we
put $j_{\true} = 1$ (\resp $j_{\false} = 1$). We put $i
\xrightarrow{\true} j_{\true}$ and $i \xrightarrow{\false}
j_{\false}$. This completes the definition of CFG for $\Delta$.

  For every action $\alpha_i$ of $\Delta$ and every valuation
  $\rho$ we put the following transitions to $\delta$.

\begin{enumerate} 
\item $\alpha_i$ is $v:=\Expr$: put
  $\tuple{\tuple{i,\rho},skip,\tuple{j,\rho[v \leftarrow
      \llbracket\Expr\rrbracket_{\rho}]}}$, where $i\rightarrow j$;\label{translation:first}
\item $\alpha_i$ is \Guess $v\in \Expr$ \with $\Test$: for every every
  $e \in \llbracket\Expr\rrbracket_{\rho}$ such that
  $\llbracket\Test\rrbracket_{\rho[v\leftarrow e]}$ is true put
  $\tuple{\tuple{i,\rho},skip,\tuple{j,\rho[v \leftarrow e]}}$, where
  $i\rightarrow j$;
\item $\alpha_i$ is $\inc(\vec{f}(\Expr))$: let $e =
  \llbracket\Expr\rrbracket_{\rho}$. Put
  $\tuple{\tuple{i,\rho},inc(c_{f(e)}),\tuple{j,\rho}}$, where
  $i\rightarrow j$;
\item $\alpha_i$ is $\dec(\vec{f}(\Expr))$: let $e =
  \llbracket\Expr\rrbracket_{\rho}$. Put
  $\tuple{\tuple{i,\rho},dec(c_{f(e)}),\tuple{j,\rho}}$, where
  $i\rightarrow j$;\label{translation:fourth}
\item $\alpha_i$ is $\Reject$:  put
  $\tuple{\tuple{i,\rho},skip,\tuple{i,\rho}}$;\label{translation:reject}
\item $\alpha_i$ is $\IIf\ \Test\ \IThen\ \alpha^{*}_{i_{\true}}\ \IElse\
  \alpha'^{*}_{i_{\false}}\ \IEndIf$: if $\llbracket\Test\rrbracket_{\rho}$ is true
 then put $\tuple{\tuple{i,\rho},skip,\tuple{i_{\true},\rho}}$, otherwise
put $\tuple{\tuple{i,\rho},skip,\tuple{i_{\false},\rho}}$, where
  $i\xrightarrow{\true} j_{\true}$ and $i\xrightarrow{\false} j_{\false}$;   
\item $\alpha_i$ is $\IIf\ \Test\ \IThen\ \alpha^{*}_{i_{\true}} \IEndIf$
or $\alpha_i$ is $\IWhile\ \Test\ \IDo\ \alpha^{*}_{i_{\true}}\IEndWhile$:
  if $\llbracket\Test\rrbracket_{\rho}$ is true then put
  $\tuple{\tuple{i,\rho},skip,\tuple{i_{\true},\rho}}$, otherwise put
  $\tuple{\tuple{i,\rho},skip,\tuple{j,\rho}}$, where
  $i\xrightarrow{\true} j_{\true}$ and $i\xrightarrow{\false} j_{\false}$. 
\end{enumerate}
Transitions defined above correspond to semantics of actions from
Figures~\ref{fig:sem-first} and~\ref{fig:sem-second}. The only
exception is for $\Reject$, which does not have any computable
effect. Instead, it is modelled as a transition of $M_H$ that loops,
and therefore cannot lead to any accepting configuration.

For a counter valuation $\vartheta$ define
$\mathit{vectorise}(\vartheta)$ as vector $\vec{n}_{\vartheta}$
indexed by the set $\{c_{\vec{f}(e)}\mid\text{ $\vec{f}\in \Vect$ and
  $e\in \mathrm{Dom}(\vec{f})$}\}$ such that
$\vec{n}_{\vartheta}(c_{\vec{f}(e)}) = \vartheta(\vec{f})(e)$.

The statement to be proven is: for every sequence of action
$\alpha^{*} \subseteq \Delta$ starting in a line $i$ and finishing in
a line $j$, and every configurations $\tuple{\rho,\vartheta}$ and
$\tuple{\rho',\vartheta'}$ of HMCA $H$ we have
$\sem{\tuple{\rho,\vartheta}}{\alpha^{*}}{\tuple{\rho',\vartheta'}}$
if and only if MCA $M_H$ transits from configuration
$\tuple{\tuple{i,\rho},\vec{n}_{\vartheta}}$ to
$\tuple{\tuple{j,\rho'},\vec{n}_{\vartheta'}}$. The proof proceeds by
a standard induction on the structure of $\alpha^*$.

From the inductive statement we conclude that
$\sem{\tuple{\rho_I,\vartheta_I}}{\Delta}{\tuple{\rho_F,\vartheta_F}}$
if and only if $\tuple{\tuple{1,\rho_I},\vec{n}_{\vartheta_I}}$
transits to $\tuple{\tuple{1,\rho_F},\vec{n}_{\vartheta_F}}$. Here
$\tuple{\rho_I,\vartheta_I}$ is the starting configuration of $H$ and
$\tuple{\rho_F,\vartheta_F}$ is an accepting configuration of $H$.
\Ie $\rho_I$ is the starting state of $H$,
$\vartheta_I(c_{\vec{f}(e)}) = \vec{0}$ for all $\vec{f} \in \Vect$
and $e\in \Type(\vec{f})$, $\rho_F\in P_F$ and $\vartheta_F$ is any
counter valuation satisfying $\vartheta_F(c_{\vec{f}(e)}) = 0$ for every
$\tuple{\vec{f},e}\in E$.  Therefore existence of an accepting run of
HMCA $H$ is equivalent to existence of an accepting run of MCA $M_H$.
\end{proof}

\subsection{Global consistency}\label{sec:global}
In this section we check global consistency of a~formula by checking
for a~given (locally consistent with the formula) frame $\Frame=
\tuple{\Kings,\ST,\Xi,\Sigma,\bar{f}}$ whether there exists a normal
structure that fits $\Frame$.  Intuitively, this is like solving
a~jigsaw puzzle: the frame gives us a~set of shapes (star types) and
we have to decide if it is possible to put pieces of these shapes
together to build a~complete picture (a~structure). Here a~piece of
a~given shape has some number of tabs (essential types) of two kinds:
tabs corresponding to invertible essential types must be connected to
a~matching tab in a~matching piece; tabs corresponding to
non-invertible essential types must be simply connected to a~matching
piece.

Figure~\ref{fig:hmca} shows a~high-level multicounter automaton
$H_{\Frame}$ that solves this
problem. The automaton guesses one by one the sequence of elements of
the structure as they appear in the order $\Lin$. In each iteration of
the transition $\Delta$ one element of the structure is guessed.

During the construction we have to make sure that several conditions
are satisfied. One of them is that all royal types from $\Kings$ are
used exactly once. For this reason we keep track of the set of nodes
visited (\ie guessed) so far; their 1-types are stored in variable
$\Visit$. The variable is updated in line~\ref{MFrame:updateForbidden}
and used in line~\ref{MFrame:guessAllowed} and in the acceptance
condition. Another condition is that the constructed structure is
normal, in particular it satisfies Condition~\ref{def:normal-second}
in Definition~\ref{def:normal-structure}. Therefore every time the
automaton guesses (the star type of) a~new element in line
\ref{MFrame:guessAllowed}, it guesses it from the set
\[
\begin{array}{rl}
\lefteqn{\Allow(\Visit) = 
\{\sigma \in \ST\mid \pi(\sigma)\in \Kings\setminus \Visit\}\;\cup} \\
&\{\sigma \in \ST \mid \forall \pi\in \Visit \left(\pi\not\in \Kings \Rightarrow 
  \exists\tau\in\Xi. \tpf{\tau} = \pi \wedge (x\Lin{y})\in \tau \wedge
  \tps{\tau} = \pi(\sigma) \right)\}.
\end{array}
\]
This way we formally forbid guessing a~royal type more than once or
violating Condition~\ref{def:normal-second} in
Definition~\ref{def:normal-structure}.  Note that here $\Allow:
2^{\Partial{\ST}}\to 2^{\Partial{\ST}}$ is a unary operator in the
language of HMCA.

Intuitively, during solving our jigsaw puzzle, we have to represent
somehow the border of the partial picture constructed so far. This is
stored in the vector $\UPending$. It is indexed with partial shapes
(formally, partial star types) because we are not interested in parts
of pieces that are already connected to the picture; we are only
interested in parts of pieces that belong to the border. The cut
vector says for each partial shape how many pieces of this shape
belong to the border. For technical reasons this vector also counts
empty shapes (which correspond to pieces that are fully connected to
the picture); these counters are then ignored in the acceptance
condition of the constructed automaton.

Formally, we define a difference type of a node \wrt to another node
in a structure.  In the analogy with jigsaw puzzle, the difference
type of $e_1$ \wrt $e$ describes the tabs of the piece $e_1$ that
belong to the border of the constructed picture just before adding the
piece~$e$.  Figure~\ref{fig:difftype} shows an example of a difference
type.
\begin{figure}
  \begin{center}
  \includegraphics{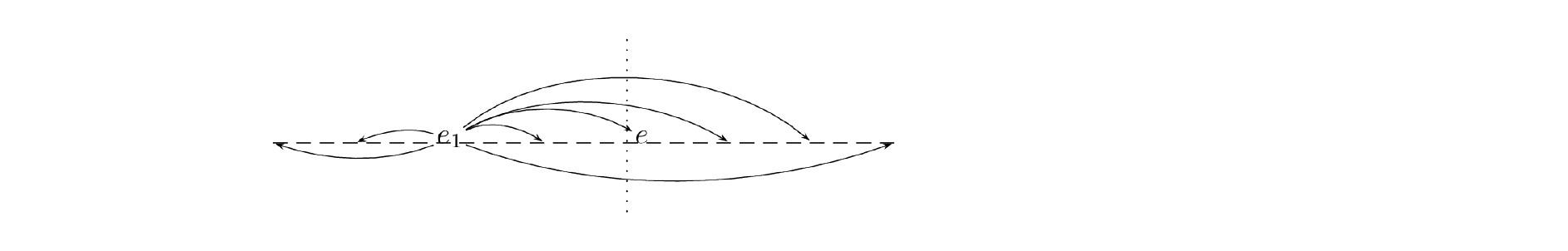}
  \end{center}
  \caption{Difference type of $e_1$ w.r.t. $e$ contains the arrows that cross the dotted line.}
  \label{fig:difftype}
\end{figure}

\begin{definition}[Difference type of $e'$ \wrt $e$ in a structure
  $\str A$]
  Let $\str A\in \ClassO$ and $e',e\in \str A$ be elements satisfying
  $e' <^{\str A} e$. Let $\sigma$ be the star type of $e'$ and let
  $\{\tau_i\}^k_{i=1}$ be essential types emitted from $e'$ and
  accepted by nodes $\Lin^{\str A}$ than $e$. The partial star type $\sigma
  \setminus \{\tau_i\}^k_{i=1}$ is called the \emph{difference type of
    $e'$ \wrt $e$ in $\str A$}.
\end{definition}

\begin{figure}[p]
\small
\begin{algorithmic}[1]
  \renewcommand{\algorithmicrequire}{\textbf{Initial
      Configuration}\xspace}
  \renewcommand{\algorithmicendif}{\textbf{endif}\xspace}
  \renewcommand{\algorithmicendwhile}{\textbf{endwhile}\xspace}
  \REQUIRE Initial state is $\rho_I$ such that $\rho_I(\sigma_c) =
  \sigma_{\textit{first}}$,  $\rho_I(\Visit) =
  \emptyset$ and the value of $\rho_I$ on remaining variables is
  arbitrary (but fixed). Initial counter valuation assigns $0$ to all
  counters.
  \renewcommand{\algorithmicrequire}{\textbf{Accepting
      Configurations}\xspace} \REQUIRE The set of accepting states
  $P_F$ consists of all states $\rho_F$ satisfying $\Kings \subseteq
  \rho_F(\Visit)$. Accepting counter valuations are defined by the set
  $E = \{\tuple{\UPending,\sigma}\mid\text{ $\sigma\in\Partial{\ST}$
    is non-empty}\}$.
  \renewcommand{\algorithmicrequire}{\textbf{Transition $\Delta$}\xspace}
  \REQUIRE

\IF{$\sigma_c = \bot$}\label{MFrame:rejectTest} 
\STATE \Reject\label{MFrame:reject}
\ENDIF
\STATE $\Visit := \Visit\cup \{\pi(\sigma_c)\}$\label{MFrame:updateForbidden}

\STATE $\sigma := \UpArrow{\sigma_c}$\label{MFrame:uparrow}
\WHILE {$\sigma \neq \emptyset$}\label{MFrame:matchWithLess-firstLine}
\STATE $\tau := \mathit{choose}(\sigma)$\label{MFrame:choose}
\STATE $\sigma := \sigma - \tau$

 \IF {$\tau$ is an invertible essential type}
 \STATE \Guess $\sigma_u \in \Partial{\ST}$ \with $\tau^{-1}\in \DownArrow{\sigma_u}$\label{MFrame:matchWithLess-firstGuess}
\ELSE 
\STATE \Guess $\sigma_u \in \Partial{\ST}$ \with $\pi(\sigma_u) = \tps{\tau}$\label{MFrame:matchWithLess-secondGuess}
\ENDIF

\STATE $\dec(\UPending[\DownArrow{\sigma_u}])$\label{MFrame:matchWithLess-firstDecrement}

\STATE $\inc(\MPending[\DownArrow{\sigma_u} - \tau^{-1}])$\label{MFrame:matchWithLess-firstIncrement}

\ENDWHILE\label{MFrame:matchWithLess-lastLine}

\STATE \Guess $\mathit{anotherIteration} \in \{\TRUE, \FALSE\}$
\WHILE{$\mathit{anotherIteration}$}\label{MFrame:matchLessWithCurrent-firstLine}
\STATE \Guess $\sigma_g
\in \Partial{\ST}$\label{MFrame:matchLessWithCurrent-guessSigmaG}
\STATE \Guess $\tau\in \DownArrow{\sigma_g}$ \with $\tps{\tau} =
\pi(\sigma_c)$ \AND ($\tau$ is non-invertible essential
type)\label{MFrame:matchLessWithCurrent-guessTau}

\STATE $\dec(\UPending[\DownArrow{\sigma_g}])$\label{MFrame:matchLessWithCurrent-dec} 
\STATE $\inc(\MPending[\DownArrow{\sigma_g} - \tau])$\label{MFrame:matchLessWithCurrent-inc}
\STATE \Guess $\mathit{anotherIteration} \in \{\TRUE, \FALSE\}$
\ENDWHILE\label{MFrame:matchLessWithCurrent-lastLine}

\STATE $\inc(\UPending[\DownArrow{\sigma_c}])$\label{MFrame:Inc}
\STATE $\tau_{\Succ} := \tau_{\Succ}({\sigma_c})$\label{MFrame:guessNext-firstLine} 
\IF {$\tau_{\Succ} = \bot$} 
\STATE $\sigma_c := \bot$
\ELSE
\STATE \Guess $\sigma_c\in \Allow(\Visit)$  \with $(\tau_{\Succ})^{-1} \in \sigma_c$ \label{MFrame:guessAllowed}
\ENDIF\label{MFrame:guessNext-lastLine}
\STATE \Guess $\mathit{anotherIteration} \in \{\TRUE, \FALSE\}$\label{MFrame:lastloop-firstline}
 \WHILE {$\mathit{anotherIteration}$}
\STATE \Guess $\sigma_g \in \Partial{\ST}$
\STATE $\dec(\MPending[\DownArrow{\sigma_g}])$ 
\STATE $\inc(\UPending[\DownArrow{\sigma_g}])$
\STATE \Guess $\mathit{anotherIteration} \in \{\TRUE, \FALSE\}$
\ENDWHILE\label{MFrame:lastloop-lastline}

\end{algorithmic}
\caption{A high-level multicounter automaton $H_{\Frame}$
  corresponding to a~frame $\Frame$.}
\label{fig:hmca}
\end{figure}

\begin{definition}
  Let $A\in \ClassO$ and $e\in \str A$.  The \emph{cut} at point $e$
  in $\str A$ is the vector $\UPending_{e}$ of natural numbers indexed
  by star types from $\Partial{\ST}$ such that
  \[\UPending_{e}[\sigma] = |\{e_1\in\str A\mid\text{difference type
    of $e_1$ \wrt $e$ in $\str A$ is $\sigma$}\}|.\]
\end{definition} \bigskip

\noindent Final states of the automaton are states where all kings from $\Kings$
are visited. Intuitively, final counter valuations are valuations
where the border of the picture is empty, \ie the counters
corresponding to non-empty shapes must be zeroed. Note that the
remaining counters, corresponding to empty shapes, store the numbers
of internal pieces of the picture with that particular shape (the
partial shapes of these pieces are empty because all tabs are already
connected to other pieces) and contain non-zero values.

Intuitively, one iteration of the transition $\Delta$ works as
follows. Before the transition starts we take a~piece $e$ of shape
$\sigma_c$; at the end of the iteration the piece is connected to the
picture constructed so far and the shape of the next piece is guessed.
The new border of the picture is computed in two loops in
lines~\ref{MFrame:matchWithLess-firstLine}--\ref{MFrame:matchLessWithCurrent-lastLine}. Before
these loops, in line~\ref{MFrame:uparrow}, we compute the tabs of $e$ that
must be connected to the picture.
In the loop in
lines~\ref{MFrame:matchWithLess-firstLine}--\ref{MFrame:matchWithLess-lastLine}
we take one by one these tabs $\tau$ and connect them to the picture:
first the shape $\sigma_u$ of a~matching piece~$e'$ in the picture is
guessed and then the border is updated: in
line~\ref{MFrame:matchWithLess-firstDecrement} the shape of $e'$ is
removed from the border and in
line~\ref{MFrame:matchWithLess-firstIncrement} the updated shape is
added to the new border. In order to avoid confusion between the old
and the new border, the new one is stored in a~separate vector
$\MPending$.  The loop in
lines~\ref{MFrame:matchLessWithCurrent-firstLine}--\ref{MFrame:matchLessWithCurrent-lastLine}
connects in a~similar way the tabs of the second kind (non-invertible
essential types) from the picture to the piece $e$.  Then remaining
tabs of~$e$ are added to the border in line~\ref{MFrame:Inc}.
Finally, the shape of the next piece is guessed and the new border is
added to the old one in
lines~\ref{MFrame:lastloop-firstline}--\ref{MFrame:lastloop-lastline}.

Formally, the set of finite-domain variables of the automaton $H_{\Frame}$ is
$\{\sigma_c,\Visit,\sigma,\tau,\sigma_u,$ $\sigma_g,\tau_{\Succ},
\mathit{anotherIteration}\}$ and there are two arrays of counters
$\UPending$ and $\MPending$.  The $\Type$ function is defined as
follows.  The type of $\sigma_c$ is $\ST \cup \{\bot\}$ (that is,
formally, $\Type(\sigma_{c}) = \ST \cup \{\bot\}$); variables
$\sigma$, $\sigma_{u}$ and $\sigma_{g}$ are of type $\Partial{\ST}$;
variable $\tau$ has type $\tau(\ST)$ and $\tau_{\Succ}$ has type
$\tau(\ST) \cup \{\bot\}$. The type of $\Visit$ is $2^{\pi(ST)}$ (that
is, this variable ranges over sets of 1-types). Finally, the type of
vectors of counters $\UPending$ and $\MPending$ is $\Partial{\ST}
\rightarrow \Nat$.

For a star type $\sigma$ define $\UpArrow{\sigma} = \{\tau\in
\sigma\mid (y<x) \in\tau\}$ and $\DownArrow{\sigma} = \{\tau\in
\sigma\mid (x<y) \in\tau\}$. Intuitively $\UpArrow{\sigma}$ (\resp
$\DownArrow{\sigma}$) denotes the set of tabs that are connected to
(\resp remain in the border of) the picture when adding a piece of
shape $\sigma$. For a star type $\sigma$ from $\ST$ define
$\tau_{\Succ}(\sigma)$ as the only 2-type $\tau$ such that
$\Succ(x,y)\in\tau$ (intuitively, this is the tab that connects a
piece of shape $\sigma$ to the next piece in the picture), or the
special value $\bot$ if $\sigma$ is the star type of the greatest element
in a structure. We also define $\mathit{choose}(\sigma)$ in line
~\ref{MFrame:choose} to be an arbitrary 2-type $\tau$ such that
$\tau\in \sigma$.

In each iteration of the transition $\Delta$ the automaton guesses one
element $e$ of the structure and assigns star type $\sigma_c$ to
it. Intuitively, at the beginning of an iteration the $\UPending$
vector stores the cut at $e$.  Technically, it is the sum of
$\UPending$ and $\MPending$ that stores the cut, but we may assume
that the $\MPending$ vector is zeroed (see Lemma~\ref{lem:emptyVec}
below). The $\MPending$ vector stores the information about changes in
cuts between the current and the next element.  The value of
$\UPending$ is updated in two loops in
lines~\ref{MFrame:matchWithLess-firstLine}--\ref{MFrame:matchLessWithCurrent-lastLine}. During
the computation some counters from $\UPending$ are decremented, and
some counters from $\MPending$ --- incremented. Decrementation of a
counter corresponds to establishing a 2-type between $e$ and some $e'$
smaller than $e$ (this is done in the loop in
lines~\ref{MFrame:matchWithLess-firstLine}--\ref{MFrame:matchWithLess-lastLine}),
or between $e'$ and $e$ (loop in
lines~\ref{MFrame:matchLessWithCurrent-firstLine}--\ref{MFrame:matchLessWithCurrent-lastLine}). In
order not to establish multiple 2-types between the same pair of
nodes, we remove the difference type of $e'$ \wrt $e$ from $\UPending$
and store it in $\MPending$. When the second loop
(lines~\ref{MFrame:matchLessWithCurrent-firstLine}--\ref{MFrame:matchLessWithCurrent-lastLine})
finishes its execution the initial value of $\UPending$ vector for
next node is the sum of the values of updated vectors $\UPending$ and
$\MPending$ in line~\ref{MFrame:matchLessWithCurrent-lastLine}.

In
lines~\ref{MFrame:guessNext-firstLine}--\ref{MFrame:guessNext-lastLine}
we guess the star type of the next node, or the special value $\bot$
in case we are already at the largest node of the structure.  Finally,
the loop in
lines~\ref{MFrame:lastloop-firstline}--\ref{MFrame:lastloop-lastline}
may move the content of vector $\MPending$ to $\UPending$. We may
assume (by Lemma~\ref{lem:emptyVec}) that the entire content is
actually moved.  Formally, the correspondence between a~frame $\Frame$
and HMCA $H_{\Frame}$ is captured by the following proposition.
\begin{proposition}\label{prop:sat-c2}
  Let $\Frame = \tuple{\Kings,\ST,\Xi,\Sigma,\bar{f}}$ be a frame.
  The automaton $H_{\Frame}$ is non-empty if and only if there exists
  a structure $\str M\in \ClassO$ that fits $\Frame$.
\end{proposition}

To prove Proposition~\ref{prop:sat-c2}, we first show (in
Lemma~\ref{lemma:sat-c1}) that if a~structure fits a~frame then the
frame induces a non-empty HMCA. Then (in Lemma~\ref{lemma:sat-c2}) we
show that if there exists a~frame for which the induced HMCA is
non-empty, then we can construct a~structure that fits the frame.

\begin{lemma}\label{lemma:sat-c1}
  Let $\Frame = \tuple{\Kings,\ST,\Xi,\Sigma,\bar{f}}$ be a frame. If there
  exists a normal structure $\str M\in \ClassO$ that
  fits $\Frame$ then $H_{\Frame}$ is non-empty.
\end{lemma}
\begin{proof}
  We will show that $H_{\Frame}$ is non-empty by presenting its
  accepting run. Let $e_1,\ldots,e_k$ be nodes of $\str M$ sequenced
  in order $\Lin^{\str M}$. Define  $\Visit_i = \{\tp{\str M}(e_j)\mid j<i\}$,
  for $i\in \underline{k}$. Recall that  $\UPending_{e_i}$ is the cut at point
  $e_i$ in $\str A$. The accepting run of $H_{\Frame}$ is
  \[\tuple{\rho_1,\vartheta_1},\ldots, \tuple{\rho_k,\vartheta_k},
  \tuple{\rho_{k+1},\vartheta_{k+1}}\] where, for $i\in
  \underline{k}$, $\rho_i(\sigma_c) = \stp{\str A}(e_i)$,
   $\rho_i(\Visit) = \Visit_i$ and the value of
  $\rho_i$ on other variables is arbitrary. Moreover
  $\vartheta_i(\UPending) = \UPending_{e_i}$ and
  $\vartheta_i(\MPending) = \vec{0}$.

  Define the last configuration of the above sequence, \ie
  $\tuple{\rho_{k+1},\vartheta_{k+1}}$, as $\rho_{k+1}(\sigma_c) =
  \bot$,  $\rho_{k+1}(\Visit) =
  \{\tp{\str M}(e_j)\mid j\leq k\}$, $\vartheta_{k+1}(\MPending) =
  \vec{0}$ and $\vartheta_{k+1}(\UPending) = \UPending_{\mathrm{Acc}}$, 
where
\[ \UPending_{\mathrm{Acc}}[\sigma] = \begin{cases} |\{e_i\mid\text{
    $1\leq i\leq k$ and $\tp{\str M}(e_1) = \pi(\sigma)$}\}|
  & \text{if $\sigma$ is empty} \\
  0 & \text{otherwise.}
\end{cases}\] By inspecting Figure~\ref{fig:hmca} we may observe that
$\tuple{\rho_{k+1},\vartheta_{k+1}}$ is indeed an accepting
configuration of $H_{\Frame}$. Next, we show that
$\tuple{\rho_1,\vartheta_1}$ is the initial configuration of
$H_{\Frame}$. Indeed, observe that $\sigma_{\mathit{first}}$ is $\stp{\str
  A}(e_1)$, and $\Visit_1 = \emptyset$. Moreover,
$\UPending_{e_1} = \vec{0}$.

To show that the specified run is indeed accepting, assume that MCA
$H_{\Frame}$ is in state $\tuple{\rho_{i},\vartheta_{i}}$, where $i\in
\underline{k}$. We will show that it may transit to state
$\tuple{\rho_{i+1},\vartheta_{i+1}}$.

Clearly, star type $\stp{\str M}(e_i)$ is different from the special
value $\bot$. Thus the test in line~\ref{MFrame:rejectTest} fails and
computation does not reject in line~\ref{MFrame:reject}.  Then, in
line~\ref{MFrame:updateForbidden} the set of visited 1-types is
updated accordingly. This way we obtain $\Visit_{i+1}$.  Next, in
lines~\ref{MFrame:matchWithLess-firstLine}--\ref{MFrame:matchLessWithCurrent-lastLine}
the value of variable $\UPending$ is transformed from
$\vartheta_{i}(\UPending)$ to $\vartheta_{i+1}(\UPending)$ in two
loops.
  Now we explain both loops in more detail.

  The loop in lines~\ref{MFrame:matchWithLess-firstLine}
  --\ref{MFrame:matchWithLess-lastLine} is responsible for verifying
  satisfaction of requirements imposed by $\UpArrow{(\stp{\str
      M}(e_i))}$, \ie for checking that essential 2-types that $e_i$
  wants to emit to smaller nodes can be accepted. The loop handles
  each 2-type $\tau \in \UpArrow{(\stp{\str M}(e_i))}$ in a single
  iteration. Let $e'$ be the unique element of $\str M$ such that
  $\tp{\str M}(e_i,e') = \tau$. There are two cases. First, if $\tau$
  is an invertible essential type then the star type $\stp{\str
    M}(e')$ contains $\tau^{-1}$.  Let $\sigma_u$ be the difference
  type of $e'$ \wrt $e_i$ in $\str M$. Since $e' \Lin^{\str M} e_i$,
  the star type $\sigma_u$ also contains $\tau^{-1}$. Thus $\sigma_u$
  can be guessed in line~\ref{MFrame:matchWithLess-firstGuess}.  The
  second case is when $\tau$ is a non-invertible essential type.  Then
  no information about connection with $e_i$ is stored in $\stp{\str
    M}(e')$.  In line~\ref{MFrame:matchWithLess-secondGuess} we guess
  $\sigma_u$, the difference type of $e'$ \wrt $e_i$ in $\str M$. In
  both cases, by definition of $\sigma_u$ we have
  $\UPending_{e_i}[\sigma_u] > 0$. Thus decrementing
  $\UPending[\sigma_u]$ in
  line~\ref{MFrame:matchWithLess-firstDecrement} will not reject. To
  remember that some edges emitted by $e'$ must be accepted by nodes
  greater than $e_i$, we then increment $\MPending[\sigma_u -
  \tau^{-1}]$. Note that when $\tau$ is a non-invertible essential
  type then $\sigma_u - \tau^{-1} = \sigma_u$.

  The loop in
  lines~\ref{MFrame:matchLessWithCurrent-firstLine}--\ref{MFrame:matchLessWithCurrent-lastLine}
  is responsible for updating the arrays $\UPending$ and $\MPending$
  to reflect connections of nodes $e'$ smaller than $e_i$ to $e_i$ by
  non-invertible essential types.  In each iteration the transition
  guesses a difference type of some node $e'$ \wrt $e_i$ in $\str M$
  and assigns it to $\sigma_g$.  Then it guesses a 2-type $\tau$ that
  $e'$ wants to emit to $e_i$
  (lines~\ref{MFrame:matchLessWithCurrent-guessSigmaG}
  and~\ref{MFrame:matchLessWithCurrent-guessTau}). To reflect that the
  connection indeed exists, $\UPending[\sigma_g]$ is decremented, type
  $\tau$ removed from $\sigma_g$ and vector $\MPending[\sigma_g -
  \tau]$ is incremented in lines~\ref{MFrame:matchLessWithCurrent-dec}
  and~\ref{MFrame:matchLessWithCurrent-inc}.  Finally, when the loop
  is over we increment $\UPending[\DownArrow{\sigma_c}]$ to reflect
  that the edges emitted by $e_i$ to greater nodes have to be matched
  later on, provided that $i< k$. When $i=k$, node $e_i$ is the last
  node of the structure, and it must accept all essential types not
  accepted by smaller nodes.

  In
  lines~\ref{MFrame:guessNext-firstLine}--\ref{MFrame:guessNext-lastLine}
  we guess the star type of the next node, or the special value $\bot$
  in case $i=k$. The next value of $\sigma_c$ must be such that it
  accepts invertible essential type $\tau_{\Succ}$ emitted by $e_i$.
  Note that $\Succ(x,y)\in\tau_{\Succ}$, which verifies that $e_{i+1}$
  is indeed the successor of $e_i$ \wrt~$\Lin^{\str M}$. Moreover,
  $\stp{\str M}(e_{i+1}) \in \Allow(\Visit_{i+1})$ as either $e_{i+1}$
  is a king in $\str M$ or it is not a king and can be connected to
  every smaller non-royal node by a silent 2-type. The latter property holds
  because $\str M$ is normal. Thus $\stp{\str M}(e_{i+1})$ can be
  guessed as the new $\sigma_c$.

  Finally, the loop in
  lines~\ref{MFrame:lastloop-firstline}--\ref{MFrame:lastloop-lastline}
  may move the content of vector $\MPending$ to $\UPending$. We may
  assume (see Lemma~\ref{lem:emptyVec} below) that the entire content
  is actually moved. Then, vector $\UPending$ obtains the value of
  $\UPending_{e_{i+1}}$.
\end{proof}

The proof of the other direction, that is the construction of
a~structure from a~run of $H_{\Frame}$, is more complicated and we
need some additional lemmas. First, observe that at the end of the
transition of $H_{\Frame}$
(lines~\ref{MFrame:lastloop-firstline}--\ref{MFrame:lastloop-lastline})
there is a loop that non-deterministically chooses a partial star-type
$\sigma_g$, decrements $\MPending[\sigma_g]$ and increments
$\UPending[\sigma_g]$. \Wlg we may assume that this procedure erases
the entire vector $\MPending$, and thus $\MPending_i = \vec{0}$ for
all $i\in \{1,\ldots k\}$. Thus we have the following lemma.
\begin{lemma}\label{lem:emptyVec}
  Let $\Frame$ be a frame such that multicounter automaton
  $H_{\Frame}$ is non-empty.  Then there exists an accepting run
  $\tuple{\rho_1,\vartheta_1},\ldots, \tuple{\rho_k,\vartheta_k}$ of
  $H_{\Frame}$ such that $\vartheta_i(\MPending) = \vec{0}$ for every
  $i\in \underline{k}$.
\end{lemma}
\begin{proof}
  Assume that $H_{\Frame}$ is non-empty and let
  $\tuple{\rho'_1,\vartheta'_1},\ldots, \tuple{\rho'_k,\vartheta'_k}$
  be its accepting run. For $i\in \underline{k}$ let $\rho_i =
  \rho'_i$, and $\vartheta_i(\UPending) = \vartheta'_i(\UPending) +
  \vartheta'_i(\MPending)$, $\vartheta_i(\MPending) = \vec{0}$. We now
  argue that the run $\tuple{\rho_1,\vartheta_1},\ldots,
  \tuple{\rho_k,\vartheta_k}$ is accepting. Consider a configuration
  $\tuple{\rho'_i,\vartheta'_i}$, for $i\in \underline{k-1}$. By
  assumption $H_{\Frame}$ may transit from
  $\tuple{\rho'_i,\vartheta'_i}$ to
  $\tuple{\rho'_{i+1},\vartheta'_{i+1}}$. We will show that
  $H_{\Frame}$ may transit from $\tuple{\rho_i,\vartheta_i}$ to
  $\tuple{\rho_{i+1},\vartheta_{i+1}}$. Execution of $H_{\Frame}$'s
  transition can be broken into two parts. First,
  lines~\ref{MFrame:rejectTest}--\ref{MFrame:guessNext-lastLine} and
  second,
  lines~\ref{MFrame:lastloop-firstline}--\ref{MFrame:lastloop-lastline}.
  
  Execution of the first part on configuration
  $\tuple{\rho'_i,\vartheta'_i}$ may be mimicked on configuration
  $\tuple{\rho_i,\vartheta_i}$ for two reasons. First, the
  execution doesn't depend on value of vector $\MPending$ (\ie the
  first part of the execution does not use decreasing actions on
  $\MPending$). Second we have $\vartheta_i(\UPending) \geq
  \vartheta'_i(\UPending)$, and since the execution does not reject on
  smaller values of $\UPending$, it won't reject on higher values of
  $\UPending$.

  In the second part of the execution
  (lines~\ref{MFrame:lastloop-firstline}--\ref{MFrame:lastloop-lastline})
  we may conduct non-deterministic guesses in such a way that the
  entire vector $\MPending$ is reset. Thus $H_{\Frame}$ may transit
  from $\tuple{\rho_i,\vartheta_i}$ to
  $\tuple{\rho_{i+1},\vartheta_{i+1}}$. Since the selection of $i$ was
  arbitrary, we conclude that all transitions in sequence
  $\tuple{\rho_1,\vartheta_1},\ldots, \tuple{\rho_k,\vartheta_k}$ can
  be conducted. Since $\tuple{\rho'_1,\vartheta'_1}$ is the starting
  configuration of $H_{\Frame}$, we have $\vartheta'_1(\MPending) =
  \vec{0}$. Furthermore, $\vartheta_1(\UPending) =
  \vartheta'_1(\UPending)$, $\vartheta_1(\MPending) = \vec{0}$ and
  $\rho_1 = \rho'_1$. Therefore $\tuple{\rho_1,\vartheta_1} =
  \tuple{\rho'_1,\vartheta'_1}$ and $\tuple{\rho_1,\vartheta_1}$ is
  the starting state of $H_{\Frame}$. Moreover, as
  $\tuple{\rho'_k,\vartheta'_k}$ is an accepting state, we have
  $\vartheta'_k(\MPending) = \vec{0}$, which leads to a~conclusion that
  $\tuple{\rho_k,\vartheta_k}$ is also an accepting state. Therefore
  the specified run is accepting, and it satisfies requirements
  imposed by the lemma.
\end{proof}

By a~\emph{partial structure} we mean a~(usually not complete) graph
whose nodes are labelled with 1-types and edges are labelled with
2-types such that whenever an edge $\tuple{e_1,e_2}$ is labelled with
$\tau$ then $e_1$ is labelled with $\tpf{\tau}$ and $e_2$ is labelled
with $\tps{\tau}$. Intuitively, partial structures are graphs that can
be extended to (full) relational structures by adding some edges.
In the following lemma we construct a partial structure from a run of
the automaton $H_\Frame$. Each iteration of the transition $\Delta$
corresponds to an extension of the structure constructed so far with
edges connecting the element guessed in the iteration. Here we are
interested only in edges that are labelled with essential types.

\begin{lemma}\label{lem:essential-substructure-exists}
  Let $\tuple{\rho_1,\vartheta_1},\ldots, \tuple{\rho_k,\vartheta_k},
  \tuple{\rho_{k+1},\vartheta_{k+1}}$ be an accepting run of
  $H_{\Frame}$ such that $\vartheta_i(\MPending) = \vec{0}$ for every
  $i\in \underline{k+1}$. There exists a partial structure $\str
  E$ with nodes $e_1,\ldots,e_k$ such that
\begin{itemize}
\item $\vartheta_i(\UPending)$ is the cut at point $e_i$ in $\str E$,
  for $i\in \underline{k}$,
\item the star type of $e_i$ in $\str E$ (\ie $\stp{\str E}(e_i)$) is
  $\rho_i(\sigma_c)$ for $i\in \underline{k}$, and
\item $\str E\models\Succ(e_i,e_{i+1})$, for
  $i\in \underline{k-1}$.
\end{itemize}
\end{lemma}
\begin{proof}
  We say that a node $e_i$ \emph{wants} to emit an edge of type $\tau$
  if $\tau\in \rho_i(\sigma_c)$ and in the partial structure
  being constructed no edge $\tau$ emitted from $e_i$ is yet present.
  We will construct the partial structure $\str E$ in $k$ steps. We
  start with an empty partial structure $\str E_0$. At the beginning
  of each step $i$, for $i\in \underline{k}$, we have a partial structure
  $\str E_{i-1}$ with nodes $e_1,\ldots e_{i-1}$ and some essential
  types established between these nodes. In step $i$ we extend $\str
  E_{i-1}$ with node $e_i$ to form $\str E_i$. During the construction
  we will maintain the following invariants:
\begin{itemize}
\item During step $i$ all essential types $\tau_{\mathit{pctr}}$ emitted
  from $e_i$ to nodes $e_j$, where $j<i$, are established. These
  essential types are defined by $\UpArrow{(\rho_i(\sigma_c))}$.
\item For any essential type $\tau_{\mathit{brdr}}$, just before step $i$,
  vector $\vartheta_i(\UPending)$ captures the number of edges
  $\tau_{\mathit{brdr}}$ that ``want'' to be emitted from nodes $e_j$,
  where $j<i$, and accepted by nodes $e_{j'}$, where $j' \geq i$. The
  number is $\sum_{\{\sigma\mid \tau_{\mathit{brdr}} \in
    \sigma\}}\left(\vartheta_i(\UPending)\right)(\sigma)$.  Some of
  these edges will be accepted by $e_i$, and some by nodes with  larger
  indices.
\end{itemize}
The first invariant trivially holds before step $1$. Since
$\vartheta_1(\UPending) = \vec{0}$, as $\tuple{\rho_1,\vartheta_1}$ is
the starting configuration of $H_{\Frame}$, the second invariant also
holds. Now, assuming that both invariants hold and that $\str E_{i-1}$
is constructed, we construct the structure $\str E_{i}$. This is done
in one transition of $H_{\Frame}$, and after the transition the
invariants are preserved.

The loop in
lines~\ref{MFrame:matchWithLess-firstLine}--\ref{MFrame:matchWithLess-lastLine}
is used for preserving the first invariant, that is, for ensuring that
all essential types $\tau_{\mathit{pctr}}\in
\UpArrow{(\rho_i(\sigma_c))}$ are established. In the loop each such
type $\tau_{\mathit{pctr}}$ is considered one by one.

There are two cases. If $\tau_{\mathit{pctr}}$ is an invertible
essential type then there must be a node~$e_j$ with $j<i$, that wants
to emit an edge $(\tau_{\mathit{pctr}})^{-1}$. Therefore, there must
be a partial star-type $\sigma_u$ such that
$(\tau_{\mathit{pctr}})^{-1} \in \sigma_u$ and
$\vartheta_{i}(\UPending)(\sigma_u) > 0$. We guess this type in
line~\ref{MFrame:matchWithLess-firstGuess}.  In the second case, if
$\tau_{\mathit{pctr}}$ is a non-invertible essential type then this
fact is not reflected in the star type of a node that will accept
it. Therefore, we only need to make sure that a node $e_j$, with $j<i$
and with 1-type $\tps{\tau_{\mathit{pctr}}}$ exists. Again,
$\vartheta_i(\UPending)$ contains sufficient data to verify this.  In
line~\ref{MFrame:matchWithLess-secondGuess} a star-type $\sigma_u$ of
such a node $e_j$ is guessed.  Next, in both cases,
line~\ref{MFrame:matchWithLess-firstDecrement} verifies that the guess
is correct and reflects that an edge between $e_i$ and $e_j$ is
established. The partial star type $\sigma_u -
(\tau_{\mathit{pctr}})^{-1}$ may be non-empty, \ie there may be other
essential types that want to be emitted from $e_j$. We store this
information by increasing $\vartheta_{i}(\MPending)(\sigma_u -
(\tau_{\mathit{pctr}})^{-1})$ in
line~\ref{MFrame:matchWithLess-firstIncrement}. Note that we use a
separate vector $\MPending$ as we must avoid assigning more than one
essential type between the same pair of nodes $e_i$ and~$e_j$. For
this reason we index $\UPending$ and $\MPending$ by partial star types
instead of essential types.

To preserve the second invariant, we look at loops in
lines~\ref{MFrame:matchLessWithCurrent-firstLine}--\ref{MFrame:matchLessWithCurrent-lastLine}
and~\ref{MFrame:lastloop-firstline}--\ref{MFrame:lastloop-lastline}.
The loop in
lines~\ref{MFrame:matchLessWithCurrent-firstLine}--\ref{MFrame:matchLessWithCurrent-lastLine}
is used  for verifying that some essential types that nodes $e_j$
(for $j<i$) want to emit to larger nodes, are accepted by $e_i$. Note that
the number of iterations made by the loop is not determined, unlike
the case of the loop considered previously. In
lines~\ref{MFrame:matchLessWithCurrent-guessSigmaG}
and~\ref{MFrame:matchLessWithCurrent-guessTau} a partial star-type
$\sigma_g$ and an essential type $\tau_{\mathit{brdr}} \in \sigma_g$ are
guessed. Note that $\tau_{\mathit{brdr}}$ cannot be invertible, as such
types were considered in the previous loop. After verifying the
correctness of the guess by decrementing an appropriate counter, we
store the information about remaining essential type that $e_j$ wants to
emit by incrementing $\vartheta_{i}(\MPending)(\sigma_g -
(\tau_{\mathit{brdr}}))$.

When the execution reaches line~\ref{MFrame:Inc}, all essential types
between $e_i$ and smaller nodes are established. Since $\rho_i(\sigma_c)$
is
$\UpArrow{(\rho_i(\sigma_c))} \cup \DownArrow{(\rho_i(\sigma_c))}$,
node $e_i$ still wants to emit essential types from
$\DownArrow{(\rho_i(\sigma_c))}$. To reflect this the action in
line~\ref{MFrame:Inc} increases an appropriate counter.

In
lines~\ref{MFrame:guessNext-firstLine}--\ref{MFrame:guessNext-lastLine}
we guess the star type of the next node, or the special value $\bot$
in case $i=k$. If $i<k$ then the next value of $\sigma_c$ must be such
that it accepts invertible essential type $\tau_{\Succ}$ emitted by
$e_i$. This ensures that nodes $e_1,\ldots,e_k$ are guessed in order.

Finally, the loop in
lines~\ref{MFrame:lastloop-firstline}--\ref{MFrame:lastloop-lastline}
may move the content of vector $\MPending$ to $\UPending$, and  by
assumption of the lemma it indeed moves the entire content of
$\MPending$ to $\UPending$, ensuring that the second invariant is
satisfied just before step $i+1$.

Put $\str E = \str E_k$ and consider the partial structure $\str E$.
Since $\tuple{\rho_{k+1},\vartheta_{k+1}}$ is an accepting
configuration, we have $\vartheta_{k+1}(\UPending)(\sigma) = 0$, for
every non-empty star type $\sigma$. Thus $\sum_{\{\sigma\mid
  \tau_{\mathit{brdr}} \in
  \sigma\}}\left(\vartheta_{k+1}(\UPending)\right)(\sigma) = 0$, for
every essential type $\tau_{\mathit{brdr}}$. By the second invariant we
conclude that for all $i\in \underline{k}$, all essential types
$\tau\in\DownArrow{( \rho_i(\sigma_c))}$ are emitted from $e_i$ and
accepted by some $e_j$, where $j>i$ and $j\leq k$. This means that
$\vartheta_i(\UPending)$ is the cut at point $e_i$ in $\str E$ (first
condition of the lemma). As guaranteed by first invariant, in step $i$
we emitted from $e_i$ all types $\tau\in\UpArrow{(\rho_i(\sigma_c))}$.
Since $\rho_i(\sigma_c)$ is $\UpArrow{(\rho_i(\sigma_c))} \cup
\DownArrow{(\rho_i(\sigma_c))}$, the node $e_i$ in $\str E_k$ does not
want to emit any essential type. This means that $e_i$ realises the star
type $\stp{\str E}(e_i) = \rho_i(\sigma_c)$ (second condition of the
lemma). Finally, as already observed, in
lines~\ref{MFrame:guessNext-firstLine}--\ref{MFrame:guessNext-lastLine}
the transition of $H_{\Frame}$ ensured that nodes of $\str E$ are
guessed in order (third condition of the lemma).
\end{proof}

In previous lemma we constructed a partial structure with all edges
labelled with essential types. This structure is still not a complete
graph because edges labelled with silent types are missing. In the
following lemma we add these edges.

\begin{lemma}\label{lemma:sat-c2}
  Let $\Frame = \tuple{\Kings,\ST,\Xi,\Sigma,\bar{f}}$ be a frame.
  If $H_{\Frame}$ is non-empty then there exists a structure $\str
  M\in \ClassO$ that fits $\Frame$.
\end{lemma}
\begin{proof}
  Assume that $H_{\Frame}$ is non-empty. By Lemma~\ref{lem:emptyVec}
  there exists an accepting run of $H_{\Frame}$:
  $\tuple{\rho_1,\vartheta_1},\ldots,
  \tuple{\rho_k,\vartheta_k},\tuple{\rho_{k+1},\vartheta_{k+1}}$ where
  $\vartheta_i(\MPending) = \vec{0}$ for every $i\in \underline{k+1}$.
  By Lemma~\ref{lem:essential-substructure-exists} there exists a
  partial structure $\str E$ whose nodes sequenced in order are
  $e_1,\ldots e_k$ and $\stp{\str E}(e_i) = \rho_i(\sigma_c)$, for
  every $i\in \underline{k}$. Note that $\Type(\sigma_{c}) = \ST$ and
  therefore $\stp{\str E}(e_i)\in \ST$. Moreover, every royal type
  from $K$ is realised in $\str E$, as the acceptance condition of
  $H_{\Frame}$ requires that $\rho_{k+1}(\Visit)$ contains the set of
  kings. We now show that we may supplement~$\str E$ to a full
  relational structure $\str M\in \ClassO$ by silent types from
  $\Xi$. This will be enough to show that $\str M$ is a structure that
  fits $\Frame$.

  Take any pair of elements $\tuple{e_i,e_j}$ such that no type is
  established for this pair in $\str E$. This means that neither $e_i$
  nor $e_j$ is a king. \Wlg assume that $i<j$. Let $\pi_i = \tp{\str
    E}(e_i)$ and $\pi_j = \tp{\str E}(e_j)$. Consider the execution of
  $H_{\Frame}$ in step $j-1$. In line~\ref{MFrame:guessAllowed} of the
  transition the star type of $e_j$ is guessed in such a way that all
  already visited non-royal 1-types can be connected by a silent type
  from $\Xi$ to the 1-type $\pi(\sigma_c)$, \ie to $\pi_j$
  (cf. Condition~\ref{def:normal-second} in
  Definition~\ref{def:normal-structure}). Since $\pi_i$ is one of such
  visited types, we may find a type $\tau\in \Xi$ such that
  $\tpf{\tau} = \pi_i$, $\tps{\tau} = \pi_j$ and $(x<y)\in \tau$. We
  then establish the type $\tau$ between~$e_i$ and $e_j$, and we
  continue the procedure for remaining pairs of nodes. In this way we
  construct the structure $\str M$. It is clear from the construction
  that $\str M$ fits $\Frame$ and that $\str M\in \ClassO$.
\end{proof}

\subsection{Main theorem}

\begin{theorem}\label{thm:main}
  The finite satisfiability problem for \CTwoOneLinS is decidable.
\end{theorem}
\begin{proof}
  We may assume that the input \CTwoOneLinS formula $\varphi$ is in
  normal form (otherwise it can be brought to the normal form).  A
  non-deterministic decision procedure for the finite satisfiability
  problem guesses a frame $\Frame$ such that $\Frame\models\varphi$
  and checks if HMCA $H_{\Frame}$ is non-empty.  If so, then by
  Proposition~\ref{prop:sat-c2} we obtain a structure $\str M\in
  \ClassO$ that fits $\Frame$. Because $\str M$ fits  $\Frame$
  and $\Frame\models\varphi$, by Proposition~\ref{prop:red} we
  conclude that $\str M\models \varphi$. This shows that $\varphi$ is
  finitely satisfiable, so our procedure is sound. On the other hand,
  if $\varphi$ is finitely satisfiable then, by
  Lemma~\ref{lem:normalstr}, it has a model $\str M$ which is a normal
  structure. Again, by Proposition~\ref{prop:red} there exists a frame
  $\Frame$ such that $\str M$ fits $\Frame$ and
  $\Frame\models\varphi$. By Proposition~\ref{prop:sat-c2} we conclude
  that $H_{\Frame}$ is non-empty, so the procedure is complete.

  Note that the size of $\Frame$ is at most doubly exponential in the
  size of the formula's signature $\tuple{\Sigma,\bar{f}}$, so there
  are finitely many frames that can be guessed. Furthermore, the
  emptiness problem of HMCA $H_{\Frame}$ is decidable, as stated in
  Proposition~\ref{prop:compile}.
\end{proof}

The decidability result for \CTwoOneLin can be applied to decide 
\CTwo with one acyclic relation. Every acyclic relation can be extended to a
linear order by topological ordering, and conversely, every relation
contained in a linear order is acyclic. Thus we have the following corollary.
\begin{corollary}
  The finite satisfiability problem for \CTwoOnePrec, where $\prec$ is
  interpreted as an acyclic relation, is decidable.
\end{corollary}
\begin{proof}
  A~formula $\varphi$ is finitely satisfiable in \CTwoOnePrec if and
  only if the formula $\varphi \wedge \forall x\forall y \;\; x\prec y
  \Rightarrow x\Lin y$ is finitely satisfiable in \CTwoOneLinPlus.
\end{proof}

\section{Conclusion}
We have shown several complexity results for finite satisfiability of
two-variable logics with counting quantifiers and linear orders. In
particular we proved \NEXPTIME-completeness of the problem for
\CTwoLinSucc, VAS-completeness for \CTwoOneLin and undecidability for
\CTwoLinUndecid. The results for \CTwoOneLin extend to \CTwo with one
acyclic relation. There are still some unsolved cases, including
\CTwoLinUnknownA and \CTwoLinUnknownB.

There are lots of open problems in the area.  One of them is general
satisfiability. None of the logics considered here has the finite model
property. Our techniques rely on finiteness of the underlying
structure, so they cannot be directly applied to general
satisfiability on possibly infinite structures.  Among possible
directions for future work one can choose combination of \CTwo with
other interpreted binary relations like preorders~\cite{ManuelZ13} or
transitive relations~\cite{szwastT13}. Another possibility is to
consider \CTwo with closure operations on some relations, like
equivalence closure~\cite{KieroMPT-lics12} or deterministic transitive
closure~\cite{CharatonikKM14}.

\bibliographystyle{plain}
\bibliography{c2lin}

\end{document}